\documentclass[preprint]{elsarticle}

\usepackage{lineno}
\usepackage{hyperref}
\hypersetup{
    colorlinks = true,
    citecolor = green,
    linkcolor =red
}
\modulolinenumbers[1]

\biboptions{sort&compress}

  \usepackage[scr=boondoxo]{mathalfa}

  \DeclareFontEncoding{FML}{}{}
  \DeclareFontSubstitution{FML}{fncmi}{m}{it}
  \DeclareSymbolFont{fourierletters}{FML}{fncmi}{m}{it}
  \SetSymbolFont{fourierletters}{normal}{FML}{fncmi}{m}{it}
  \DeclareMathSymbol{\fP}{\mathalpha}{fourierletters}{`E}

  \DeclareSymbolFont{Nperm}{OML}{cmss}{bx}{it}
  \SetSymbolFont{Nperm}{bold}{OML}{cmss}{bx}{it}
  \DeclareMathSymbol{\nw}{\mathalpha}{Nperm}{`w}

  \makeatletter
  \let\save@mathaccent\mathaccent
  \newcommand*\if@single[3]{%
  \setbox0\hbox{${\mathaccent"0362{#1}}^H$}%
  \setbox2\hbox{${\mathaccent"0362{\kern0pt#1}}^H$}%
  \ifdim\ht0=\ht2 #3\else #2\fi
  }
  \newcommand*\rel@kern[1]{\kern#1\dimexpr\macc@kerna}
  \newcommand*\widebar[1]{\@ifnextchar^{{\wide@bar{#1}{0}}}{\wide@bar{#1}{1}}}
  \newcommand*\wide@bar[2]{\if@single{#1}{\wide@bar@{#1}{#2}{1}}{\wide@bar@{#1}{#2}{2}}}
  \newcommand*\wide@bar@[3]{%
  \begingroup
  \def\mathaccent##1##2{%
  \let\mathaccent\save@mathaccent
  \if#32 \let\macc@nucleus\first@char \fi
  \setbox\z@\hbox{$\macc@style{\macc@nucleus}_{}$}%
  \setbox\tw@\hbox{$\macc@style{\macc@nucleus}{}_{}$}%
  \dimen@\wd\tw@
  \advance\dimen@-\wd\z@
  \divide\dimen@ 3
  \@tempdima\wd\tw@
  \advance\@tempdima-\scriptspace
  \divide\@tempdima 10
  \advance\dimen@-\@tempdima
  \ifdim\dimen@>\z@ \dimen@0pt\fi
  \rel@kern{0.6}\kern-\dimen@
  \if#31
  \overline{\rel@kern{-0.6}\kern\dimen@\macc@nucleus\rel@kern{0.4}\kern\dimen@}%
  \advance\dimen@0.4\dimexpr\macc@kerna
  \let\final@kern#2%
  \ifdim\dimen@<\z@ \let\final@kern1\fi
  \if\final@kern1 \kern-\dimen@\fi
  \else
  \overline{\rel@kern{-0.6}\kern\dimen@#1}%
  \fi
  }%
  \macc@depth\@ne
  \let\math@bgroup\@empty \let\math@egroup\macc@set@skewchar
  \mathsurround\z@ \frozen@everymath{\mathgroup\macc@group\relax}%
  \macc@set@skewchar\relax
  \let\mathaccentV\macc@nested@a
  \if#31
  \macc@nested@a\relax111{#1}%
  \else
  \def\gobble@till@marker##1\endmarker{}%
  \futurelet\first@char\gobble@till@marker#1\endmarker
  \ifcat\noexpand\first@char A\else
  \def\first@char{}%
  \fi
  \macc@nested@a\relax111{\first@char}%
  \fi
  \endgroup
  }
  \makeatother

  \usepackage{graphicx}

  \usepackage{amssymb,amsmath,yhmath, amsthm}
  \newtheorem{theorem}{Theorem}[section]
  \usepackage{bm}

  \usepackage{algorithm}
  \usepackage{algpseudocode}

  \usepackage[T1]{fontenc}
  \usepackage{tikz}

  \newcommand{\tikzcircle}[2][red,fill=red]{\tikz[baseline=-0.5ex]\draw[#1,radius=#2] (0,0) circle ;}%
  \newcommand{\stable}{\tikzcircle[black,fill=black]{0.75pt}}

  \newcommand{\cp}{\mathbin{\tikz [x=1.4ex,y=1.4ex,line width=.1ex] \draw (-0.5,-0.5) -- (0.5,0.5) (-0.5,0.5) -- (0.5,-0.5);}}%

  \usepackage{enumitem}
  \usepackage{pifont}
  \usepackage{epstopdf}
  \usepackage{xr}
  \usepackage{hyperref}
  \hypersetup{
  colorlinks = true,
  citecolor = blue,
  linkcolor = blue
  }

  \numberwithin{equation}{section}

  \newcommand{\physF}{\boldsymbol{\mathfrak{f}}} %
  \newcommand{\physB}{\mathscr{b}} 
  \newcommand{\physE}{{\mathpzc{E}}} 
  \newcommand{\physh}{\mathscr{h}} 
  \newcommand{\physe}{\mathbscr{e}} 
  \newcommand{\physf}{\boldsymbol{\mathfrak{f}}} 
  \newcommand{\physm}{\mathscr{m}} %
  \newcommand{\physu}{\mathbscr{u}}
  \newcommand{\physw}{\mathscr{w}}
  \newcommand{\rtple}[1]{\boldsymbol{#1}}

  \renewcommand{\u}[1]{\boldsymbol{#1}}

  \renewcommand{\t}[1]{\tilde{#1}}
  \renewcommand{\b}[1]{\mathbb{#1}}
  \renewcommand{\c}[1]{\mathcal{#1}}

  \newcommand{\pr}[1]{\left( #1 \right)}

  \newcommand{\Dro}{\dot{\rho}}
  \newcommand{\DDro}{\ddot{\rho}}

  \newcommand{\Drop}{\Dro^{+}}
  \newcommand{\Drom}{\Dro^{-}}
  \newcommand{\Dl}{\dot{l}}

  \newcommand{\DF}{\dot{F}}

  \newcommand{\Dst}{\mathcal{D}^{\stable}}
  \newcommand{\Dnt}{\mathcal{D}^{\odot}}
  
  \newcommand{\conf}{\boldsymbol{\kappa}}
  \newcommand{\confp}{\tilde{\boldsymbol{\kappa}}}

  \DeclareMathAlphabet{\mathvz}{OT1}{LinuxBiolinumT-OsF}{m}{sl}

  \DeclareMathAlphabet{\mathlib}{OT1}{LinuxLibertineT-OsF}{m}{it}
  \DeclareMathAlphabet{\mathbio}{OT1}{LinuxBiolinumT-OsF}{m}{it}

  \DeclareMathAlphabet{\mathpzc}{OT1}{pzc}{m}{it}

  \makeatletter
  \DeclareFontFamily{U}{tipa}{}
  \DeclareFontShape{U}{tipa}{m}{n}{<->tipa10}{}
  \newcommand{\arc@char}{{\usefont{U}{tipa}{m}{n}\symbol{62}}}%

  \newcommand{\arc}[1]{\mathpalette\arc@arc{#1}}

  \newcommand{\arc@arc}[2]{%
    \sbox0{$\m@th#1#2$}%
    \vbox{
      \hbox{\resizebox{\wd0}{\height}{\arc@char}}
      \nointerlineskip
      \box0
    }%
  }
  \makeatother

\begin{document}

\begin{frontmatter}

\title{Angle-independent optimal adhesion in plane peeling of thin elastic films at large surface roughnesses}


\author[address]{Weilin Deng}
\author[address]{Haneesh Kesari\corref{mycorrespondingauthor}}
\cortext[mycorrespondingauthor]{Corresponding author}
\ead{haneesh\_kesari@brown.edu}

\address[address]{School of Engineering, Brown University, Providence, RI 02912}

\numberwithin{equation}{section}

\begin{abstract}
  Adhesive peeling of a thin elastic film from a substrate is a classic problem in mechanics. However, many of the investigations on this topic to date have focused on peeling from substrates with flat surfaces. In this paper, we study the problem of peeling an elastic thin film from a rigid substrate that has periodic surface undulations. We allow for contact between the detached part of the film with the substrate. We give analytical results for computing the equilibrium force given the true peeling angle, which is the angle at which the detached part of the film leaves the substrate. When there is no contact we present explicit results for computing the true peeling angle from the substrate's profile and for determining an equilibrium state's stability solely from the substrate's surface  curvature. The general results that we derive for the case involving contact allow us to explore the regime of peeling at large surface roughnesses. Our analysis of this regime reveals that the peel-off force can be made to become independent of the peeling direction by roughening the surface. This result is in stark contrast to results from peeling on flat surfaces, where the peel-off force strongly depends on the peeling direction. Our analysis also reveals that in the large roughness regime the peel-off force achieves its theoretical upper bound, irrespective of the other particulars of the substrate's surface profile.
\end{abstract}


\end{frontmatter}


\section{Introduction}
\label{sec:Introduction}

Bumpy protrusions on the surface of the lotus leaf at the small length scales have been shown to endow the leaf with the property of super-hydrophobicity (non-wettability) at the large scales~\cite{barthlott1997purity}. 
This non-wetting property is the source of the lotus leaf's acclaimed self-cleaning ability. 
Similarly, the periodic, small-scale, wavy undulations on the skin of some sharks are thought to reduce the skin's frictional drag~\cite{wen2014biomimetic}, as shown in Figure~\ref{fig:Example}a. 
Such intriguing links between small-scale mechanical structures and large-scale physical properties, however, are not limited to biological systems. 
The small-scale topography of natural surfaces, which is generally stochastic, is often referred to as surface roughness.
When separating two contacting surfaces, the surfaces' roughness is typically thought to reduce the adhesion between them~\cite{fuller1975effect,levins1995impact,quon1999measurement,rabinovich2000adhesion}. However, there are cases in which roughness is seen to actually enhance adhesion ~\cite{briggs1977effect,kesari2010role}.
Small-scale mechanical design can not just modulate surface properties at the large-scale but, in fact, give rise to completely new phenomenon at the large scales.
For example, it was shown by Kesari \textit{et al.}~\cite{kesari2011effective,kesari2010role,deng2019depth,deng2019effect,deng2017molecular} that small-scale surface topography can give rise to the phenomenon of depth-dependent hysteresis at the large scales.

In engineering sciences, there is currently tremendous interest in creating materials with novel properties, such as materials with negative Poisson's ratios, through the use of small-scale, intricate lattice-based mechanical designs (e.g., Figure~\ref{fig:Example}b).
The recent popularity of 3D printing is, presumably, primarily responsible for galvanizing such interests.
The concept of modulating material properties at the large scale through incorporating mechanical designs at the small-scale---in contrast to, say, through the use of chemical or metallurgical treatments---has been of interest, in fact, for the past several decades in the composites and the mathematical homogenization communities~\cite{michel1999effective,ericksen2012homogenization,castaneda2002second}.
However, the focus in engineering and mathematics has primarily been on modulating bulk material properties, such as elastic stiffnesses and thermal conductivities.
Needless to add, surface physical properties, such as adhesion, friction, etc., are no less important than bulk properties, and, as can be gleaned from the  examples previously mentioned,  small-scale mechanical designs can lead to not just modulation but the emergence of completely new surface properties at the large scale. 
Therefore, a reason behind the engineering community's narrow focus could be that the mathematical theories that connect  large-scale surface physical properties to small-scale mechanical designs are not as well developed as the ones that connect bulk physical properties to small-scale mechanical designs. 
In this paper, we focus on the surface property of adhesion and present a mathematical theory that connects the large-scale force needed to peel off a thin elastic film from a rigid substrate to the substrate's small-scale surface topography.

Adhesion is one of the most important surface physical properties. (See~\cite{deng2019effect,popov2017strength,ciavarella2019role} for discussions on this aspect.)
Adhesive peeling of a thin film from a substrate  is ubiquitous in biological and engineering systems and draws growing attention due to its many applications (e.g., Figure~\ref{fig:Example}c).
Despite its long history, the topic of thin film peeling continues to reveal new and interesting phenomena. We review some of the recent studies on this topic in \S\ref{sec:LitReview}.

In this paper we study the mechanics of peeling in a plane (2-dimensional) problem.
Specifically, we study the peeling of a thin elastic film from a rigid substrate whose surface is nominally flat with superimposed periodic undulations in a single direction.
To be more precise, the substrate's topography  is, respectively, invariant and periodic in two orthogonal directions that lie in the substrate's nominal surface plane. 
The invariant and periodic directions are shown marked as $\physe_{3}$ and $\physe_{1}$, respectively, in Figure~\ref{fig:Schematic}.
The applied tractions as well as all fields in the film also do not vary in the $\physe_3$ direction.
We assume that the film's de-adherence only takes place at the peeling front, which is a straight line parallel to $\physe_3$. %
We find that in the regime where the substrate's surface roughness is large,
the peel-off force becomes independent of the direction in which the film's free end is being pulled.
This result is in stark contrast to the results from the seminal analysis of Rivlin~\cite{rivlin1944the} and Kendall~\cite{kendall1975thin}, who studied the classical problem of plane peeling of a thin elastic film from a flat surface.
In that classical problem the peel-off force strongly depends on the direction of peeling.
We  also find that this angle-independent peel-off force's magnitude, in fact, equals the maximum value that is possible for the peel-off force during the plane peeling of a thin elastic film from a rigid substrate, irrespective of the details of the substrate's profile and the direction of peeling.
The aforementioned results are especially significant considering that
nowadays it is routinely possible to create highly regular small-scale topographies on engineering surfaces through the use of micro-fabrication and 3D printing technologies.

We make the following assumptions in our problem.
We assume that initially  the thin film is perfectly adhered to the substrate, i.e., with no gaps between it and the substrate's surface. 
Over the adhered region we do not allow for (tangential) slip between the thin film and the substrate's surface.
We model adhesive interactions between the film and the substrate using the JKR theory~\cite{johnson1971surface}, as per which during peeling the system's energy increases linearly with the area of the region of the film-substrate interface that comes de-adhered.
We consider a force-controlled peeling process.
The force is applied to the thin film's free end and its magnitude and direction can vary in a fairly general way with time, as long as the force  is tensile in nature.
We term the angle that the applied force makes with the direction that is parallel to $\physe_1$ and points away from adhered portion of the film the \textit{nominal peeling angle}, or just the peeling angle for short (Figure~\ref{fig:Schematic}).
We take the peeling process to be quasi-static in nature and ignore all inelastic and inertial effects.
We do use the notion of time, however, but only so that we can speak about the \emph{sequence} of  events that take place in our peeling experiment.
By a peeling experiment we mean a sequence of peeling angles and force magnitudes applied to the thin film, which are infinitesimally different from each other.
   
We take the thin film to be composed of a linear elastic material and to be of vanishing thickness.
Consequently, we assume that the thin film's elastic strain energy only depends on its stretching deformations. Specifically, we ignore any ``bending energy'' in the thin film.
The deformations related to stretching, however, can be of finite magnitude.
Despite these and the many other assumptions that we introduce over the course of the paper, our analysis reveals important and interesting new mechanics. 
This is possible, we believe, because we retain the important feature of contact in our problem.
That is, we allow for the detached part of the thin film to come into contact with the substrate's surface and assume that such contact is non-adhesive and frictionless.
By retaining this important feature of contact we are able to investigate the regime where the substrate's surface roughness is large.
Additionally, it is in this regime where the peel-off force becomes independent of the peeling angle and the force's magnitude becomes equal to its maximum value.

\begin{figure}[t!]
    \centering
    \includegraphics[width=0.8\textwidth]{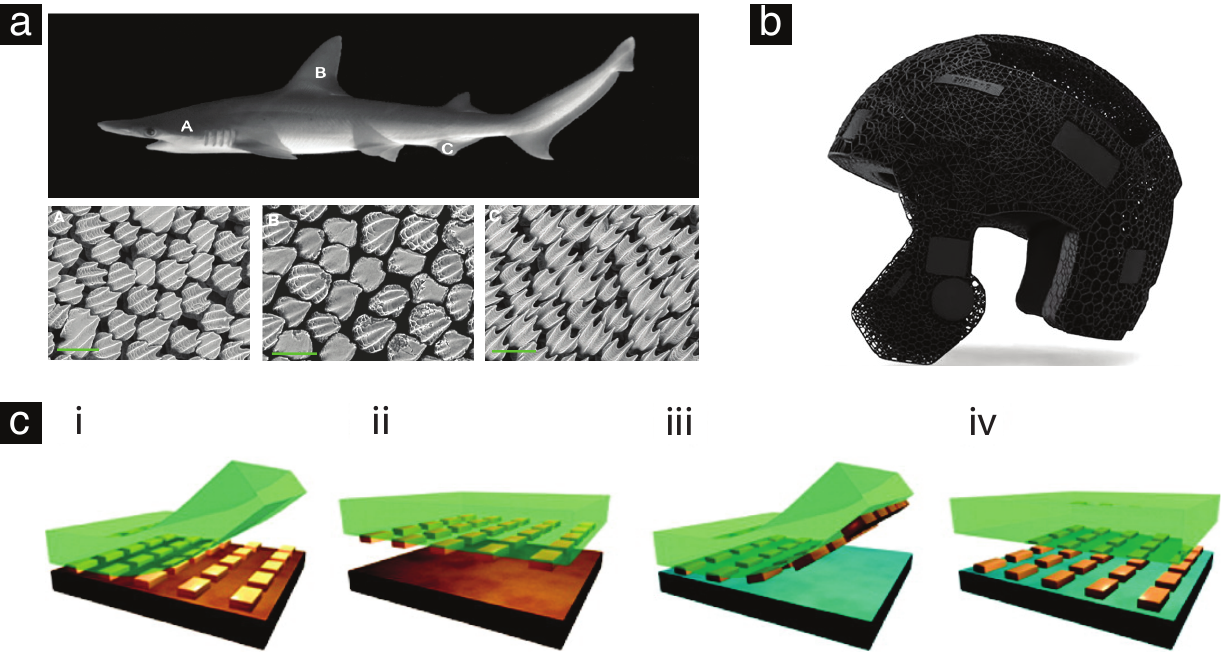}
    \caption{(a) The bonnethead shark (Sphyrna tiburo) skin surface at different body locations~\cite{wen2014biomimetic}. (b) 3D printed helmet (from~\cite{helmet}, used without permission). (c) Picking up (i--ii) and placing (iii--iv) solid objects from one substrate to another in transfer printing~\cite{feng2007competing}.}
\label{fig:Example}
\end{figure}

The outline of the paper is as follows.

We briefly review some of recent results of the mechanics of peeling in~\S\ref{sec:LitReview}.
We introduce the technical details of our problem in \S\ref{sec:model}.
Specifically, we describe the kinematics of our problem in \S\ref{sec:geometry}, and present the law that we use for modeling the thin film's de-adherence in~\S\ref{sec:energy}.
We take a \textit{configuration} in our problem to be defined by the de-adhered length, the peeling angle, and the substrate's surface profile.
The de-adhered length is defined in \S\ref{sec:geometry}, and is roughly the width of the region on the substrate's nominal surface in the $\physe_1$ direction from which the film has been detached. 
In \S\ref{subsubsec:ContactConditions} we derive the sufficiency condition, which we term the global compatibility condition, for there to be no contact between the detached part of the thin film and the substrate in a given configuration.
We present  results  for the case in which the global compatibility condition holds, i.e., in which there is no contact, in \S\ref{sec:PeelingNoContact}, and for the case in which it is violated in \S\ref{sec:PeelingWithContact}.
(When the global compatibility condition is violated the configurations that occur during a peeling experiment may or may not involve contact.)

For both cases we give explicit, analytical results for computing the equilibrium force given the \textit{true peeling angle},
which is the angle at which the detached part of the film leaves the substrate (Figure~\ref{fig:PeelingNoContactAnalysis}).
For the case in which there is no contact, we present explicit results for computing the true peeling angle from the substrate's profile, and for determining an equilibrium state's stability just from the substrate's surface curvature. 
We call a configuration along with the magnitude of the force acting on it a \textit{state}. 
For the case in which there can be contact, we present an algorithm (see Algorithm~\ref{algo:TruePeelingAngle}) that allows for numerically calculating the a configuration's true peeling angle, and determining an equilibrium state's stability.

For both cases, we also prove that there exist critical force values, such that if the magnitude of the applied force lies outside their range, the film either continuously de-adheres or adheres until the film is completely peeled off from or adhered to the substrate. 
We term these quantities the \textit{peel-initiation} and \textit{peel-off} forces.  
We determine the  peel-initiation force's minimum value as well the peel-off force's maximum value for the wide class of surface profiles that we consider in our problem and all admissible peeling angles.

We prove that in a peeling experiment in which the peeling angle is held at a fixed value that is both acute (resp. obtuse) and violates the global compatibility condition the peel-off (resp. peel-initiation) force achieves its maximum (minimum) value, irrespective of the substrate's profile.

Again considering experiments in which the peeling angle is held constant
we  show in~\S\ref{sec:Optimal} that as the surface roughness becomes large the peel-off force becomes independent of the experiment's peeling angle. This happens irrespective of the substrate's surface profile. The surface roughness becoming large is equivalent to the scenario where the lateral length scales in the substrate's surface topography becomes small. Furthermore, we also show that the magnitude of that angle-independent, peel-off force equals the maximum that is possible to be achieved. 

We conclude the  paper by discussing the effect of the thin film's elastic bending energy on the peel-off force in \S\ref{sec:discussion}.

\section{Literature review}
\label{sec:LitReview}

The mechanics of thin film peeling has been extensively studied over the decades after the seminal work of Rivlin~\cite{rivlin1944the} and Kendall~\cite{kendall1975thin}, who analyzed the peeling of a thin film from a flat surface and its importance in understanding the adhesion mechanisms in thin film-substrate systems. 
For example, Pesika \textit{et al.}~\cite{pesika2007peel} presented a peel-zone model to study the effect of the peeling angle and friction on adhesion based on the microscopic observation of the geometry of the peel zone during film detachment. Chen \textit{et al.}~\cite{chen2008pre} considered the effect of pre-stretching the film in Kendall's peeling model and found that the pre-tension significantly increases the peeling force at small peeling angles while decreasing it at large angles. Molinari and Ravichandran~\cite{molinari2008peeling} proposed a general model for the peeling of non-linearly elastic  thin films and investigated the effects of large deformations and pre-stretching. Begley \textit{et al.}~\cite{begley2013peeling} developed a finite deformation analytical model for the peeling of an elastic tape and defined an effective mixed-mode interface toughness to account for frictional sliding between the surfaces in the adhered region.
Gialamas \textit{et al.}~\cite{gialamas2014peeling} used a Dugdale-type cohesive zone to model adhesion between an incompressible neo-Hookean elastic membrane and a flat substrate and carried out  both \textit{single-} and \textit{double-sided} peeling analysis, while ignoring the membrane's bending stiffness.
Menga \textit{et al.}~\cite{menga2018multiple} investigated the periodic double-sided peeling of an elastic thin film from a deformable layer that is supported by a rigid foundation.
Kim and Aravas~\cite{kim1988elasto} performed an elasto-plastic analysis of a thin film peeling problem in which  the de-adhered part of the film is in \textit{pure-bending}.
Kinloch \textit{et al.}~\cite{kinloch1994peeling} performed an elasto-plastic analysis  of the peeling of flexible laminates and calculated the fracture energy, which included not just the energy required to break the interfacial adhesive bonds but also the energy dissipated in the plastic/viscoelastic zone at the peeling front.
Wei and Hutchinson~\cite{wei1998interface} numerically investigated the steady-state, elasto-plastic peeling of a thin film from an elastic substrate using a cohesive zone model.
Loukis \textit{et al.}~\cite{loukis1991effects} analyzed the peeling of a viscoelastic thin film from a rigid substrate and related the fracture energy and peel-off force to the peeling speed and film thickness.
Afferrante and Carbone~\cite{afferrante2016ultratough} investigated the peeling of an elastic thin film from a flat viscoelastic substrate and  gave a closed-form expression relating the peeling force to the peeling angle and the work of adhesion.

Spatial heterogeneity in the thin film's elasticity and the interface's strength can lead to significant enhancement of the effective peel-off force. The role of spatial heterogeneity in peeling was first highlighted by Kendall~\cite{kendall1975control} who carried out peeling experiments using an elastic thin  film that had alternating large and small bending stiffness regions. He varied the stiffness by either introducing reinforcements at select positions on a uniform thin film or by varying the film's thickness along its length. He peeled the film from a rigid substrate at a constant force and observed abrupt changes in the speed of the peeling front at the boundaries between the stiff and compliant regions, and an overall enhancement in the peeling force. Ghatak \textit{et al.}~\cite{ghatak2004peeling} and Chung and Chaudhury~\cite{chung2005roles} performed displacement-controlled peeling experiments between a flexible plate and an incision-patterned thin elastic layer and found that the patterns significantly enhanced the effective adhesion. They attributed the  enhancement to arrest and re-initiation of the peeling front motion  at the edges of the features in the pattern. More recently, Xia \textit{et al.}~\cite{xia2012toughening,xia2013adhesion,xia2015adhesion} experimentally and theoretically   investigated the peeling of an elastic thin film while spatially varying the bending stiffness and interface strength in it. They showed that the film's large-scale (``effective'') adhesive properties  could be significantly enhanced through the use of small-scale spatial heterogenity.

There have been several studies that investigate the effect of surface roughness on  thin film adhesion. For instance, Zhao \textit{et al.}~\cite{zhao2013improvement} simulated the peeling of a hyperelastic thin film from a rough substrate using finite elements and a cohesive zone model. They found that the peeling force could be increased by introducing a hierarchical wavy interface between the film and the substrate.
Ghatak~\cite{ghatak2014peeling} theoretically investigated the peeling of a flexible plate from an adhesive layer which was supported on a rigid substrate and had spatially varying surface topography and elastic modulus.
He found that the maximum adhesion enhancement took place when the surface height and elastic modulus varied in phase.
Peng and Chen~\cite{peng2015peeling} studied the peeling of an elastic thin film from a rigid substrate having sinusoidally varying surface topography. They computed the peeling forces for different peeling angles and surface topography parameters and found that the maximum peeling force could be significantly enhanced by increasing the substrate's surface roughness. To our knowledge there have not been any studies that consider contact between the de-adhered portion of the film and the substrate, or which study the stability of the equilibrium configurations, as we do in this current paper. The surface topography that we consider is fairly general and the regime we explore, namely where the lateral length scales in the substrate's surface topography are much smaller than the other length scales in the problem, has also not been explored within the context of thin film peeling.

    \section{Theory of wavy peeling}
    \label{sec:Theory}

    \begin{figure}[t!]
      \centering
      \includegraphics[width=0.8\textwidth]{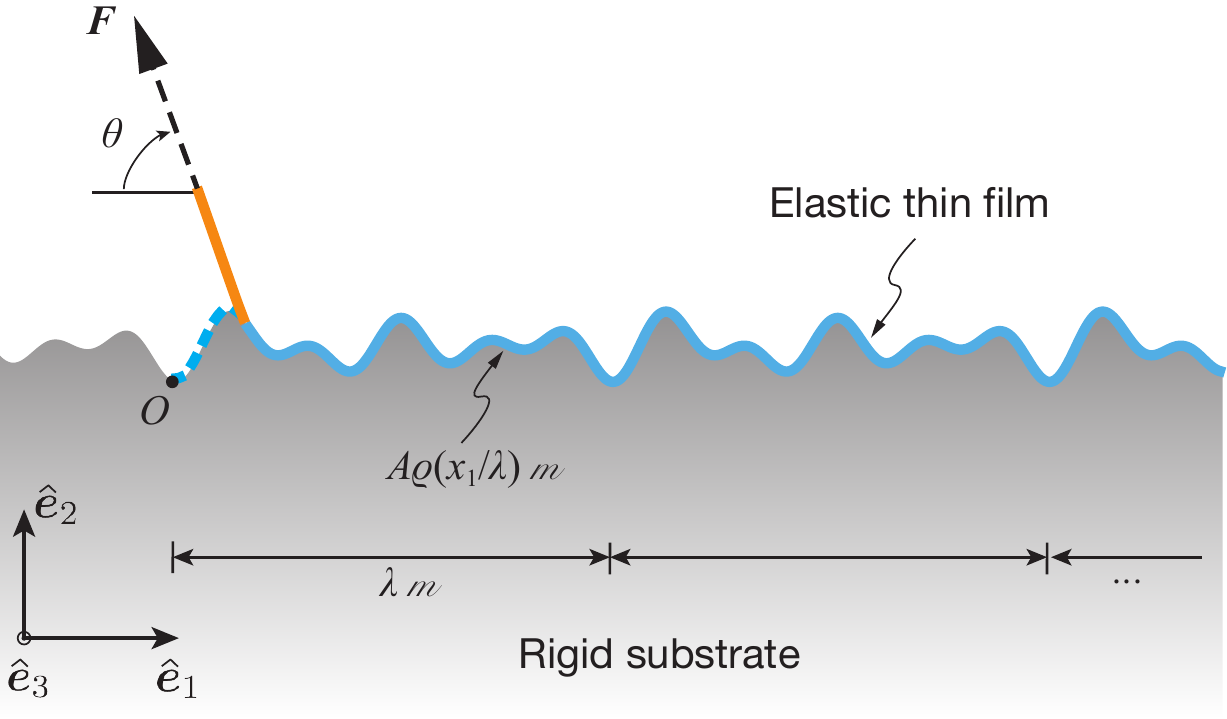}
      \caption{The schematic of peeling a thin elastic film from a rigid substrate with a periodic, wavy surface.}
      \label{fig:Schematic}
    \end{figure}

    \subsection{Model}
    \label{sec:model}
    \subsubsection{Geometry and kinematics}
    \label{sec:geometry}

    Figure~\ref{fig:Schematic} shows the geometry of our wavy peeling problem.
    Let $\b{E}$ be a three-dimensional, oriented, Hilbert space, and let the Euclidean point space $\c{E}$ be $\b{E}$'s principle homogeneous space.
    The elements of $\mathbb{E}$ have units of length.
    The physical objects in our peeling problem are contained in $\c{E}$.
    The origin of $\c{E}$, which we denote as $O$, is marked in Figure~\ref{fig:Schematic}.
    The vector set $\left(\hat{\physe}_i\right)_{i\in \mathcal{I}}$, where the index set $\mathcal{I}:=(1,2,3)$, is an orthonormal set in $\mathbb{E}$.
    That is, $\hat{\physe}_i\cdot \hat{\physe}_{j}=\delta_{ij}$, where $i$,$j\in \c{I}$, the symbol $\cdot$ denotes the inner product operator between the elements of $\b{E}$, and $\delta_{ij}$ is the Kronecker-delta symbol, which equals  unity if $i=j$ and naught otherwise.
    A typical point $X \in \c{E}$ is identified by its coordinates $\left(x_i\right)_{i\in \mathcal{I}}\in \b{R}^3$\footnote{In our work the manner in which the information about a physical quantity's units is stored  is different from how that is usually done. We model different physical quantities as vectors belonging to different vector spaces.  We store the information about a  physical quantity's units in the basis vectors we  choose for that quantity's vector space.
    For example, we may take $\hat{\mathbf{e}}_1$ to represent a motion of $1\, \mathrm{meters}$  (or $1\, \mathrm{micron}$)  in a certain fixed direction in $\mathbb{E}$. In that case  $\mathbf{e}_1$, as a consequence of its definition, would denote a motion of
    $\lambda \, \mathrm{meters}$   (resp. $\lambda \, \mathrm{microns}$) in the same direction as $\hat{\mathbf{e}}_1$.
    Thus, in our work the components of a physical  quantity with respect to the basis vectors chosen for its vector space will always be non-dimensional. For example the components of $X$ with respect to $(\mathbf{e}_i)$,  namely $\pr{x_i}_{i\in\mathcal{I}}$, are dimensionless. 
    } that are defined such that $X=O+\sum_{i\in \mathcal{I}}x_i \physe_i$.
    The vector $\sum_{i\in \mathcal{I}}x_i \physe_i\in \mathbb{E}$ is called $X$'s position vector.

    Note that the vector set $(\physe_i)_{i\in \mathcal{I}}$ is different from $(\hat{\physe}_i)_{i\in \mathcal{I}}$; it is defined as  $\physe_i :=\lambda \hat{\physe}_i$, where $\lambda > 0$, for $i=1,~2$, and $\physe_3 := \hat{\physe}_3$. The set $\pr{\physe_i}_{i\in\mathcal{I}}$ is an orthogonal, but not an orthonormal, basis for $\mathbb{E}$.
    We will be following the Einstein summation convention in this paper\footnotemark, and write expressions such as $\left(\physe_i\right)_{i\in \c{I}}$ and $\sum_{i\in \mathcal{I}}x_i \physe_i$ simply as $\left(\physe_i\right)$ and $x_i \physe_i$, respectively, and lists such as $\physe_1,~\physe_2$, and $\physe_3$ simply as $\physe_i$.

    \footnotetext{We follow the Einstein summation convention in this paper.
    If a symbol has an italicized, light-faced latin character that appears as its subscript/superscript then that subscript/superscript denotes an index and that  symbol along with that subscript/superscript denotes a component of a linear mapping.  An index that appears only once in a term is called a ``free'' index. A free index in a term denotes that the term in fact represents the tuple of terms created by varying the free index in the original term over   $\c{I}$.  An index that appears twice in a term  is called a ``repeated index''. A repeated index in a term denotes a sum of the terms that are created by varying the free index in the original term over   $\c{I}$.}

    The substrate is a rigid solid whose points' position vectors belong to the set
    \begin{linenomath}
    \begin{equation}
      \mathbb{S}=
      \left\{x_1\physe_1+x_2\physe_2+x_3\physe_3 \in \mathbb{E}~|~
      (x_1,x_2,x_3)\in \mathbb{R}^3~\text{and}~
      x_2 < \rho(x_1) :=\alpha \varrho(x_1)\right\},
      \label{eq:SubstrateGeometry}
    \end{equation}
    \end{linenomath}
    where $\alpha:=A/\lambda$, $A\ge 0$, $\varrho: \mathbb{R}\to [-1,1]$ is a surjective, 1-periodic function, and $\mathbb{R}$ is the set of real numbers.
    By 1-periodic we mean that $\varrho(x_1+1)=\varrho(x_1)$ for all $x_1\in \mathbb{R}$.
    Without loss of generality we take $\int_0^1 \varrho(x_1) \, dx_1 = 0$. 
    We also assume that $\varrho$ and its first and second derivatives, which we denote as $\dot \varrho$ and $\ddot \varrho$, respectively, are non-constant, continuous functions, i.e., $\varrho \in C^2$. 
    We call $\alpha$, $A$, $\lambda$, and $\varrho$ the substrate surface's aspect ratio, amplitude, periodicity, and profile, respectively.

    We take the thin film to be of width $\physB$, of thickness $\physh$, and to be perfectly adhered  to the substrate's surface in its initial, or reference, configuration. Here $\physB$ and $\physh$ are physical parameters and have units of length. Say the units of $\physe_i$ and $\hat{\physe}_i$ is some length $\mathscr{m}$, which, for example, may stand for meters, then we would say that the magnitude of $\physB$  (resp. $\physh$) is $b$  (resp. $h$) \textit{iff} $\physB=b\,\physm$ (resp. $\physh=h\,\physm$).
    Thus, initially the configuration of the thin film can be described as
    \begin{linenomath}
    \begin{equation}
      \mathbb{B}_{0}=\left\{ x_{1}\physe_1+\rho\left(x_{1}\right)\physe_2+x_{3}\physe_3\in \mathbb{E}\;|\;x_{1}\in\mathcal{D},\;\vert x_{3}\rvert\le b/2\right\},
      \label{eq:ThinFilmGeometry}
    \end{equation}
    \end{linenomath}
    where $\mathcal{D}:=[0,+\infty)$\footnotemark.

    \footnotetext{As is standard in continuum mechanics, we will be referring to a particular thin film material particle using the coordinates of the the spatial point that it occupied in the initial configuration. 
     That is, when we say a material particle $(x_1,x_2,x_3)$, we in fact are referring to the material particle that in the initial configuration occupied the spatial point with coordinates $(x_1,x_2,x_3)$. 
     For the sake of brevity, we will be referring to a material particle  $(x_1,x_2,x_3)$ simply as $(x_1,x_3)$, since the second co-ordinate of any spatial point that was occupied by a thin film material particle in the initial configuration is fully determined by the point's first co-ordinate  (cf.~\eqref{eq:ThinFilmGeometry}). 
     Finally, when we say ``the (thin-film) material particle $x_1$'', where $x_1\in\mathcal{D}$, we in fact mean the group of material particles $(x_1,x_3)$, where $\lvert x_3\rvert\le b/2$. 
     }

    We model the peeling process by assuming that any deformed configuration $\boldsymbol{\kappa}$ of the thin film can be described as\footnotemark
    \begin{linenomath}
		\begin{subequations}
			\begin{align}
				\mathbb{B}
				=\left\{
				\mathbscr{x}(x_1)+x_{3}\physe_3 \in \mathbb{E}\;|\;x_{1}\in\mathcal{D},\;\vert x_{3}\rvert\le b/2\right\} ,
				\label{eq:DefConf}
				\intertext{where}
				\mathbscr{x}(x_1)
				:=
				\left(x_{1}+u_{1}(x_{1})\right)\physe_1+\left(\rho(x_{1})+u_{2}(x_{1})\right)\physe_2, \label{eq:DeformationMappingFromManifoldToCurrent}
			\end{align}
        \end{subequations}
    \end{linenomath}
    $u_{1},u_{2}:\mathcal{D}\to\mathbb{R}$ and their derivatives, denoted as $\dot{u}_{1}$ and $\dot{u}_{2}$, are continuous. We further assume that  $u_{1}\left(x_{1}\right)=u_{2}\left(x_{1}\right)=0$ for all $x_{1}\ge a$, for some $a\in \mathcal{D}$, and $\dot{u}_{1}\left(x_{1}\right)=\dot{u}_{2}\left(x_{1}\right)=0$ for all $x_{1}>a$.
    We refer to  $a$ as \textit{the de-adhered length}. We call $P:=\left\{ O+a\physe_1+\rho(a)\physe_2+x_3\physe_3 \in \mathcal{E}\;|\; \lvert x_{3}\rvert\le b/2\right\}$ the peeling front and $\Gamma_{a}=\{x_1 \in \mathcal{D}~|~x_1>a\}$ the adhered region.

    \footnotetext{In standard continuum mechanics the reference and deformed configurations are taken to belong to different spaces. In contrast, here we take both the reference and deformed configurations,  $\mathbb{B}_0$ and $\mathbb{B}$, to belong to the space space, namely $\mathbb{E}$.}

    On knowing the de-adhered length $a$ the length of the film peeled from the substrate  can be computed as  $\lambda l(a\,;\rho)\,\physm$, where $l(\cdot\,;\rho):\mathcal{D}\to \mathcal{D}$  is defined by the equation
    \begin{linenomath}
    \begin{equation}
      l(a\,;\rho)=\int_{0}^{a}dx_1 \left(1+\Dro(x_1)^2\right)^{1/2}.
      \label{eq:l}
    \end{equation}
    \end{linenomath}
    We refer to $l(a\,;\rho)$ as the peeled length, and will often abbreviate it as $l$.

    \subsubsection{Evolution law for de-adherence}
    \label{sec:energy}
    In postulating the evolution law for de-adherence,  we   ignore all dynamic effects, such as inertial forces, kinetic energy, and viscoelastic behavior  in the thin film. As we mentioned in \S\ref{sec:Introduction}, we do use the notion of time, but only so that we can speak about the \emph{sequence} of  events that take place in our peeling experiment.

    In our peeling experiment the thin film de-adheres due to the application of the force $\physF$ to its free end, which initially is at $O$, the origin of $\mathcal{E}$.
    We think of  force as a linear map from $\mathbb{E}$ into a one dimensional vector space whose elements have units of energy.
    The set of all forces can, of course, be made into a vector space, which we will denote as $\mathbb{F}$.
    We use the set $\pr{\physf_i}$ as a basis for $\mathbb{F}$, where $\physf_i$ are defined such that $\physf_i\pr{\physe_j}= \mathbio{E}_{c}\, \delta_{ij}$.
    The symbol $\mathbio{E}_{c}$, which appears in the last expression, is defined to be equal to the energy $\lambda \physE \physB \physh  \physm$, where $\physE$ is the thin film's Young's modulus.
    We express $\physF$ as $ F \hat{\physf}$, where $F\ge0$ and $\hat{\physf}:=-\cos(\theta) \physf_1+\sin(\theta)\physf_2$. The angle $\theta$ is shown marked in Figure~\ref{fig:PeelingNoContactAnalysis}. It is prescribed as part of the problem's definition. We call it the \textit{nominal peeling angle}, or simply \textit{peeling angle} for short.
    We denote the position vector of the free end of the thin film, on which $\physF$ acts, as $\physu\in \mathbb{E}$,  and will often refer to this quantity   as simply the force-position-vector. Our experiment is force controlled and quasi-static.
    By force-controlled we mean that while the system's configuration is changing,  $\physF$ is held fixed.
    If our experiment were a general force controlled experiment, then we would be allowed to vary  $\physF$ once the adhered region had stopped evolving.
    However, in our experiment  we are only allowed to vary $F$ after the adhered region has stopped evolving.
    By quasi-static we mean that each variation in $F$  is of infinitesimal magnitude and is applied  all at once at a particular instance in time.

    Consider an experimentally observed configuration of the thin film $\boldsymbol{\kappa}_0^{\stable}$, in which the de-adhered length is $a-\Delta a$, the peeled length is $l-\Delta l$, the force position vector is $\physu-\Delta \physu$, and the force is $(F-\Delta F)\hat{\physf}$.
    Imagine that a force variation of $\Delta F \hat{\physF}$    is then abruptly added to the force that was previously acting on the thin film so that the new force is $F\hat{\physF}$.
    As a result, the thin film's configuration will evolve to a new configuration $\boldsymbol{\kappa}^{\stable}$,  in which the  de-adhered length is $a$, the peeled length is $l$, the force-position-vector is $\physu$, and, of course, the force is $F\hat{\physF}$.
    We postulate that the new configuration $\boldsymbol{\kappa}^{\stable}$ is the one that locally minimizes the system's total potential energy and is closest to $\boldsymbol{\kappa}_0^{\stable}$.

    To make the postulate  precise, consider a configuration $\widetilde{\boldsymbol{\kappa}}$ that is close to $\boldsymbol{\kappa}^{\stable}$ and has the same force as $\boldsymbol{\kappa}^{\stable}$, i.e., $F\hat{\physF}$.
    The de-adhered and peeled lengths, and the force-position-vector in $\widetilde{\boldsymbol{\kappa}}$ are, however, different from those in $\boldsymbol{\kappa}^{\stable}$; we denote them, respectively, as $a+\delta a$,  $l+\delta l$, and $\physu+\delta \physu$.
    Since $F\ge 0$, i.e. the peeled section of the film is in tension, the variations
    $\delta l$ and $\delta \physu$, in fact, depend on $\delta a$.
    Thus, these variations are to be interpreted as abbreviations for $\delta l(\delta a\,;a,F,\rho)$ and $\delta \physu(\delta a\, ;a,F,\rho)$, respectively.
    Let the difference in the system's potential energy between the configurations $\widetilde{\boldsymbol{\kappa}}$ and $\boldsymbol{\kappa}^{\stable}$ be $\mathbio{E}_c\, \delta E$, where $\mathbio{E}_{c}$ is a constant and has unit of energy, and we refer to $\delta E\in \mathbb{R}$ as the non-dimensional potential energy variation.
    In our model of the wavy peeling experiment, we take $\delta E$  to be a sum of  three different energy variations.
    These variations take place, respectively, in the energy stored in the interbody adhesive interactions between the thin film and the substrate, the elastic strain energy stored in the peeled part of the thin film, and, finally, the energy stored in the apparatus that maintains the constant force $F\hat{\physf}$ between the configuration $\boldsymbol{\kappa}^{\stable}$ and $\widetilde{\boldsymbol{\kappa}}$.

    We model adhesive interactions between the thin film and the substrate using the JKR theory~\cite{johnson1971surface}. According to this theory,  the formation of  an interface region of area $\delta \mathscr{A}$ lowers  the system's total potential energy, irrespective of the shape of the interface region or any other details of the experiment, reversibly by  $\physw \delta \mathscr{A}$. Here,  $\physw$ is the Dupr\'e work of adhesion~\cite[p.~30]{maugis:book}. Thus,  the variation in the potential energy stored in the interbody adhesive interactions is
    \begin{linenomath}
    \begin{equation}
      w \delta l,
      \label{eq:Es}
    \end{equation}
    \end{linenomath}
    where $w$ is defined such that $ \physE \physh w =\physw$. That is, $w$ is the non-dimensional work of adhesion.

    We assume that the peeled part of the  film is in uniform, uniaxial tension. It follows from this assumption that the strain in the peeled part of the film is also uniform, and uniaxial, and that its magnitude is equal to $F$. It also follows that the variation in the elastic strain energy stored in the thin film is
    \begin{linenomath}
    \begin{equation}
      \frac{1}{2}F^2 \delta l.
      \label{eq:Eel}
    \end{equation}
    \end{linenomath}
    Usually, the above term is supplemented by an additional term that corresponds to bending energy in the thin film~\cite{peng2015peeling}.
    We, however, ignore the film's bending energy in our model. As we state in our closing remarks, we believe that ignoring the bending energy is unlikely to significantly affect our estimates for the film's peel-off force in the type of the peeling experiments that we consider in this paper.

    The variation in the potential energy of the apparatus maintaining the constant force $F\hat{\physF}$ is $-\physF(\delta \physu)$.
    Say $\delta \physu=\delta u_j\physe_j$.
    This variation can therefore be written as
    \begin{linenomath}
    \begin{equation}
      -\boldsymbol{F}\cdot \delta \boldsymbol{u}\footnotemark,
      \label{eq:Ep}
    \end{equation}
    \end{linenomath}
    where $\boldsymbol{F}:= F \hat{\boldsymbol{F}}$,
    $\hat{\bm{F}} := (-\cos(\theta),\sin(\theta),0)$, and $\delta \boldsymbol{u}:= \pr{\delta u_i}$.
    \footnotetext{The symbol $\cdot$ in this expression denotes the inner product operator between elements of $\mathbb{R}^3$.}
    It follows from the previous discussion that,
    \begin{linenomath}
    \begin{equation}
      \delta E= w\delta l+\frac{1}{2}F^2\delta l-\boldsymbol{F}\cdot \delta \boldsymbol{u}.
      \label{eq:dE}
    \end{equation}
    \end{linenomath}
    As is the case with $\delta l$ and $\delta \rtple{u}$, the symbol $\delta E$  in~\eqref{eq:dE} is, in fact, an abbreviation for the value $\delta E(\delta a;\,a, F,\rho)$.

    \renewcommand*{\thefootnote}{\arabic{footnote}}

    Since $\delta l$ and $\delta \rtple{u}$ depend on $\delta a$, the formula for $\delta E$ given by~\eqref{eq:dE} needs to be  further refined before it can be used for determining whether or not a particular de-adhered length $a$ is in equilibrium; and, if $a$ is in equilibrium, then for determining the stability of that equilibrium. The particulars of the  requisite refinement depend on whether or not the peeled part of the film contacts the substrate. Therefore, we will be making separate refinements of~\eqref{eq:dE} for the cases of no-contact and contact in \S\ref{sec:PeelingNoContact} and \S\ref{sec:PeelingWithContact}, respectively. In both cases, however, the refinement process will involve  determining  the asymptotic dependence of $\delta l$ and $\delta \rtple{u}$ on $\delta a$ as $\delta a\to 0$, and then  using that information to  determine the asymptotic dependence of $\delta E$ on $\delta a$.

    \subsubsection{Conditions for the peeled part of thin film to not contact the substrate}
    \label{subsubsec:ContactConditions}
    Let $\boldsymbol{T}\left(x_{1}\right)$ be  $(\physe_{1}\cdot \dot{\boldsymbol{\mathbscr{x}}}(x_{1}),\physe_{2}\cdot \dot{\boldsymbol{\mathbscr{x}}}(x_{1}),0)\in \mathbb{R}^3$, where $\mathbscr{x}(x_1)$ is defined in~\eqref{eq:DeformationMappingFromManifoldToCurrent}. Because a material particle $x_1$ is de-adhered from or adhered to the substrate (depending on whether $x_1$ is less than or greater than $a$), the vector $\boldsymbol{T}(x_{1})$ will often be discontinuous at $x_{1}=a$.
    However, since we have assumed $\dot{\varrho}$, $\dot{u}_1$, and $\dot{u}_2$ to be continuous, the following right and left hand limits of $\boldsymbol{T}\left(x_{1}\right)$ as $x_{1}\to a$ are well defined:
    \begin{linenomath}
    \begin{equation*}
      \boldsymbol{T}\left(a^{\pm}\right):= \lim_{x_{1}\to a^{\pm}}\boldsymbol{T}(x_{1}).
    \end{equation*}
    \end{linenomath}
    The vector $\boldsymbol{T}\left(a^{+}\right)$ is essentially the
    tangent to the substrate's surface profile at $x_{1}=a$, and the
    vector $-\boldsymbol{T}\left(a^{-}\right)$ points in the direction
    the detached part of the thin film leaves the substrate's surface
    at the peeling front. (See Figure~\ref{fig:PeelingNoContactAnalysis}.)
    We call the angle between $\boldsymbol{T}\left(a^{-}\right)$
    and $\boldsymbol{T}\left(a^{+}\right)$ the true peeling angle $\psi(a)$.
    For the thin film to not go through the substrate immediately after de-adhering from it, it is necessary that
    \begin{linenomath}
    \begin{equation}
      \psi(a)\in[0,\pi].
      \label{eq:localComp}
    \end{equation}
    \end{linenomath}
    We refer to~\eqref{eq:localComp} as the local compatibility condition.

    It can be shown that once the local compatibility condition is satisfied, the peeled part of the thin film will not contact the substrate anywhere, irrespective of the location of the peeling front \textit{iff}
    \begin{linenomath}
      \begin{subequations}
        \begin{align}
          \theta
          &\in \left[-\tan^{-1}\pr{\Dro^{-}}, \pi - \tan^{-1}\pr{\Dro^{+}}\right],\label{eq:ThetaNotRange}
          \intertext{where}
          \Dro^{\pm}&:=\Dro\pr{a^{\pm}},\label{eq:DRhoPlusMinus}\\
          a^{\pm} &:= \arg\max\left\{\pm\, \Dro\pr{x_1}~|~ x_1 \in [0,1]\right\}.\label{eq:aPlusMinus}
        \end{align}
        \label{eq:GlobalComp}
      \end{subequations}
    \end{linenomath}
    We refer to  \eqref{eq:GlobalComp} as the global compatibility condition.

    \subsection{Peeling with no contact}
    \label{sec:PeelingNoContact}
    \begin{figure}[t!]
      \centering
      \includegraphics[width=0.7\textwidth]{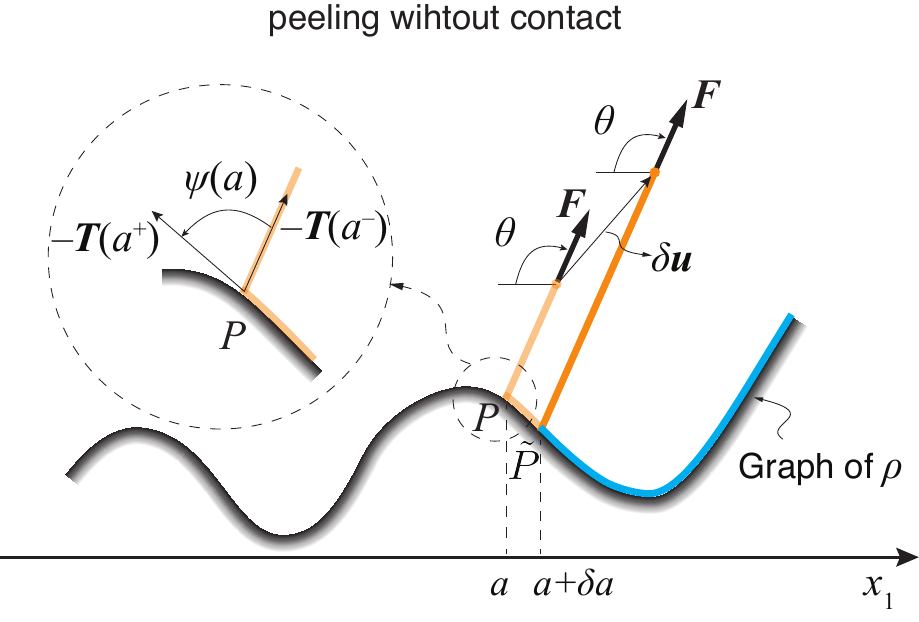}
      \caption{The schematic of peeling without contact during an infinitesimal advance of the peeling front from $P$ to $\tilde{P}$, where $\delta \bm{u}$ denotes the variation of position vector of the free end of the thin film, $\psi(a)$ is the true peeling angle.}
      \label{fig:PeelingNoContactAnalysis}
    \end{figure}

    In this section we study the case in which the local and global compatibility conditions (i.e.,~\eqref{eq:localComp} and~\eqref{eq:GlobalComp}) are satisfied, and therefore  the peeled part of the thin film does not contact the substrate anywhere.

    \subsubsection{Energy variation}
    \label{sec:EnergyVariation_NC}

    Since the peeled part of the thin film is not in contact with the substrate, it follows   that
    \begin{linenomath}
    \begin{equation}
      \delta {\bm{u}}
      = (\delta {a}, \delta \rho) + \delta {l}(1+\varepsilon)\hat{\bm{F}},
      \label{eq:du_NC}
    \end{equation}
    \end{linenomath}
    where $\varepsilon$ is the uniform, uni-axial strain in the peeled part of the thin film, and $\delta\rho:= {\rho}({a}+\delta {a}) - {\rho}({a})$.
    It follows from~\eqref{eq:l} that
    \begin{linenomath}
    \begin{equation}
        \delta l(\delta a;a, \rho)=\dot{l}(a;\rho)\delta a+\frac{1}{2}\ddot{l}(a;\rho)\pr{\delta a}^2+o\pr{\pr{\delta a}^2},\label{eq:deltalAsymExpan}
    \end{equation}
    \end{linenomath}
    where
    \begin{linenomath}
    \begin{subequations}
    \begin{align}
      \dot{l}(a;\rho)&=\pr{1+\Dro(a)^2}^{\frac{1}{2}},\label{eq:Dl}\\
      \ddot{l}(a;\rho)&=\frac{\dot{\rho}(a) \ddot{\rho}(a)}{\pr{1+\dot{\rho}(a)^2}^{\frac{1}{2}}}. \label{eq:DDl}
    \end{align}
    \label{eq:PLDerivatives}
    \end{subequations}
    \end{linenomath}
    Substituting the asymptotic expansions for $\delta \rho$ and $\delta l$ as $\delta a \to 0$ into \eqref{eq:du_NC}, and then substituting the resulting asymptotic expansion for $\delta \boldsymbol{u}$ into~\eqref{eq:Ep}, we find that
    \begin{linenomath}
    \begin{equation}
      \begin{aligned}
        -{\bm{F}} \cdot \delta {\bm{u}}
        &= {F}\left(\cos(\theta) - \sin(\theta)  \Dro({a}) - (1+\varepsilon)\dot{l}({a}) \right)\delta {a} \\
        & \qquad - \frac{1}{2} {F}\left( \sin(\theta) \DDro({a}) + (1+\varepsilon)\ddot{l}({a}) \right) (\delta {a})^2 + o\pr{(\delta {a})^2}.
      \end{aligned}
      \label{eq:Fdu_NC}
    \end{equation}
    \end{linenomath}
    The symbol $o$, that appears in~\eqref{eq:Fdu_NC} and elsewhere, is the Bachmann-Landau ``Small-Oh'' symbol. Its primary property of relevance is that  $o\pr{\pr{\delta a}^n}/(\delta a)^n\to 0$, where $n=0,~1,\ldots$, as $\delta a\to 0$.

    \begin{figure}[htbp!]
      \centering
      \includegraphics[width=\textwidth]{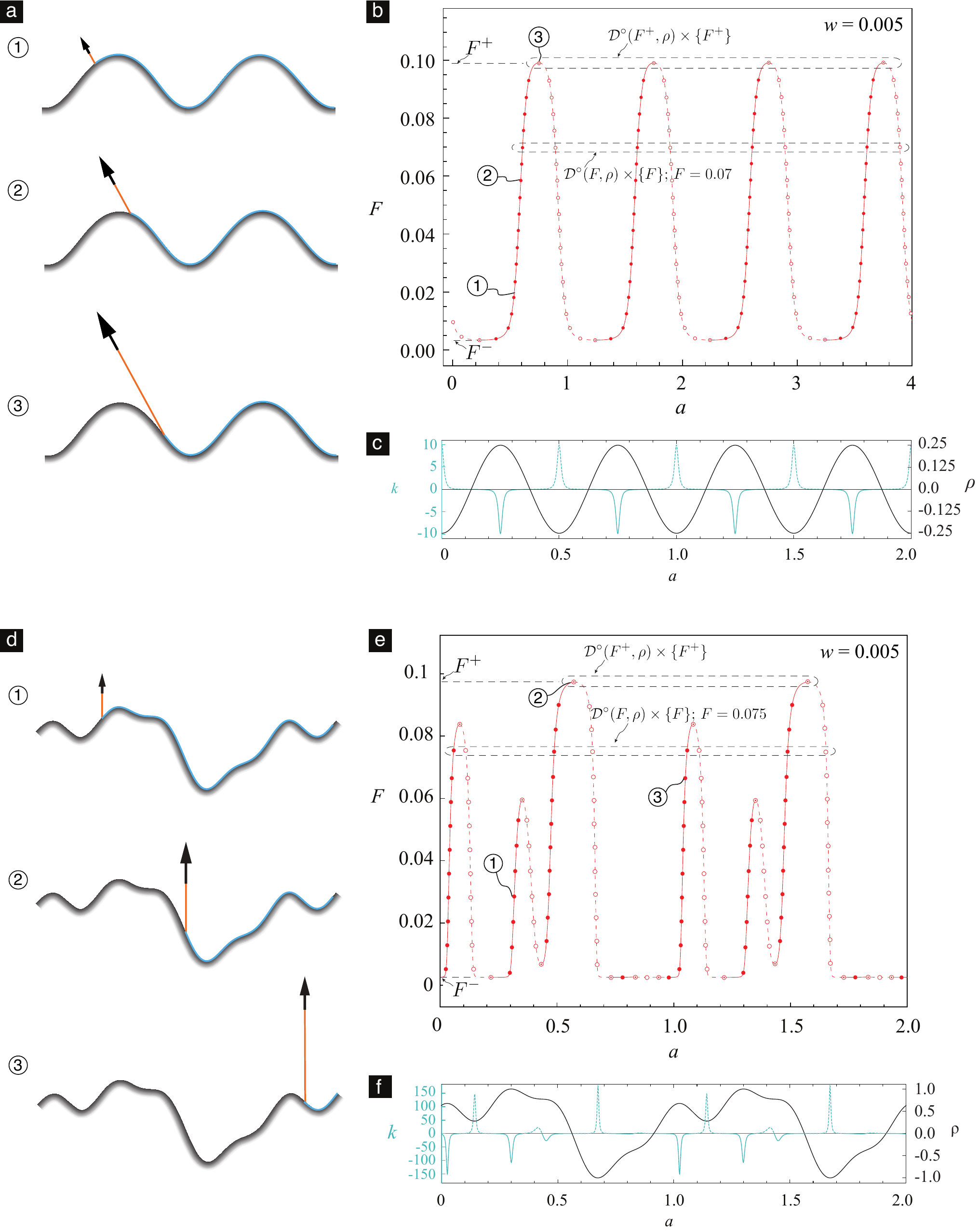}
      \caption{(a)--(c): Thin film peeling on a sinusoidal surface that does not involve contact. In this example, $A = 0.25$, $\lambda = 1.0$, $\theta = \pi/3$, and $\varrho({x}_1) = -\cos\left( 2\pi {x}_1 \right)$.
      (a) The peeling configurations corresponding to the equilibrium state as marked in (b).
      (b) The plot of the set of points $\mathcal{D}^{\circ}(F,\rho) \times \{F\}$. The peel-off force $F^+ = 0.099$, which is much greater than the peel-off force, 0.01, for peeling on a flat surface with the same nominal peeling angle.
      (c) The plot of the signed curvature and graph of $\rho$ of the sinusoidal surface.
      (d)--(f): Thin film peeling on a complicated surface that does not involve contact. In this example, $A = 1.0$, $\lambda = 1.0$, $\theta = \pi/2$, $\varrho(x_1) = \pr{\sin(2\pi x_1) + \cos(4\pi x_1 + \pi/4)/2 + \sin(8\pi x_1 + \pi/3) + 0.256}/1.315$.
      (d) The peeling configurations corresponding to the equilibrium state as marked in (e).
      (e) The plot of the set of points $\mathcal{D}^{\circ}(F,\rho) \times \{F\}$. The peel-off force $F^+ = 0.97$, which is much greater than the peel-off force, 0.005, for peeling on a flat surface with the same nominal peeling angle.
      (f) The plot of the signed curvature and graph of $\rho$ of the complicated surface.
      In (b) and (e), the stable equilibrium state is marked as $\bullet$, neutral as $\odot$, and unstable as $\circ$.
      }
      \label{fig:FNonContact}
    \end{figure}

    In the current case of no contact it also follows that the true peeling angle is given by
    \begin{linenomath}
    \begin{equation}
      \psi(a) = \theta + \tan^{-1}\pr{\Dro\pr{a}}.
      \label{eq:psi_NC}
    \end{equation}
    \end{linenomath}
    By recognizing from~\eqref{eq:psi_NC}  that
    \begin{linenomath}
    \begin{equation}
      \cos(\psi(a)) \dot{l}({a};\rho)= \cos(\theta) - \sin(\theta) \Dro(a),
      \label{eq:cos_psi}
    \end{equation}
    \end{linenomath}
    we can rewrite~\eqref{eq:Fdu_NC} as
    \begin{linenomath}
    \begin{equation}
      \begin{aligned}
        -{\bm{F}} \cdot \delta {\bm{u}}
        &= {F} \left(\cos(\psi(a)) - (1+\varepsilon) \right)\dot{l}({a};\rho) \delta {a} \\
        & \qquad - \frac{1}{2} {F}\left( \sin(\theta) \DDro({a}) + (1+\varepsilon)\ddot{l}({a};\rho) \right) (\delta {a})^2 + o((\delta {a})^2).
      \end{aligned}
      \label{eq:Fdu_NC1}
    \end{equation}
    \end{linenomath}

    By substituting~\eqref{eq:Fdu_NC1} into~\eqref{eq:dE}, and noting that on account of our non-dimensionalization scheme $\varepsilon=F$, we get that
    \begin{linenomath}
      \begin{subequations}
        \begin{align}
          \delta E (\delta {a};a,{F},\rho) &=
            \delta E_{1}(a\, ;\rho, F)\delta a+
            \delta E_{2}(a\, ;\rho, F)\pr{\delta a}^2+o\pr{\pr{\delta a}^2}, \label{eq:dPi_NC}
          \intertext{where}
          \delta E_{1}(a\, ;\rho, F) &:= \left( -\frac{1}{2}{F}^2 + {F}(\cos(\psi(a))-1) + {w} \right)\dot{l}({a};\rho), \label{eq:dPida_NC} \\
          \delta E_{2}(a\, ;\rho, F) &:= - \frac{1}{2} \left( {F} \sin(\theta) \DDro({a}) + \left(\frac{{F}^2}{2}+{F}-{w}\right)\ddot{l}({a};\rho) \right). \label{eq:d2Pida2_NC}
        \end{align}
      \end{subequations}
    \end{linenomath}

  \subsubsection{Equilibrium state}
  \label{sec:EqmPeelingNoContact}

  Since we take $\rho \in C^2$, it follows that $\Dro$ and  $\dot{l}$ are continuous, and, as a consequence, that $\delta E_1(\cdot;\rho,F)$ is continuous. It then follows that for $a$ to be an equilibrium de-adhered length it is necessary that  $\delta E_{1}(a\, ;\rho, F)=0$. For a given $F$ and $\rho$ there can, however, be more than one de-adhered length that is in equilibrium. We characterize all those lengths by saying that they belong to the set
  \begin{linenomath}
  \begin{equation}
    \mathcal{D}^{\circ}(F,\rho):=\left\{a\in \mathcal{D}~|~\delta E_{1}(a\,;\rho,F)=0\right\}.
    \label{eq:EqmDomain}
  \end{equation}
  \end{linenomath}
  We plot the set of points
  \begin{linenomath}
  \[\mathcal{D}^{\circ}(F,\rho) \cp  \{F\}\footnotemark\]
  \end{linenomath}
  for various $F$ values,  for the example surfaces shown in Figures~\ref{fig:FNonContact}a and c in Figures~\ref{fig:FNonContact}b and d, respectively. As can be seen, the peeling force for peeling on wavy surface is not constant, but varies with the same periodicity of the wavy surfaces.
  It can be shown that the point set $\mathcal{D}^{\circ}(F,\rho) \cp  \{F\}$, for any admissible
  $F$, falls on the graph of the periodic function
  \begin{linenomath}
    \begin{subequations}
      \begin{align}
        F(\cdot)
        &:=\mathpzc{F}\circ \psi, \label{eq:rmF_func}
        \intertext{where $\psi$ is defined in~\eqref{eq:psi_NC} and $\mathpzc{F}: [0,\,\pi]\to\mathcal{D}$ is defined by the equation}
        \mathpzc{F}(\psi) = \cos\pr{\psi} - 1 &+ \left(\left( \cos(\psi) - 1 \right)^2 + 2{w} \right)^{1/2}.
        \label{eq:F_func}
      \end{align}
      \label{eq:FasFuncofa}
    \end{subequations}
  \end{linenomath}
  We only consider cases in which the work of adhesion $w$ is non-negative.
  It therefore follows from~\eqref{eq:F_func} that $\mathpzc{F}(\psi)$ is always non-negative.
  The graph of $\mathcal{F}$ for different  $w$ values is shown in Figure~\ref{fig:F-theta}.
  As can be seen, $\mathcal{F}$ is a strictly decreasing function whose value at any admissible $\psi$ increases with $w$.

  Kendall analyzed the peeling of a thin film on a flat smooth surface~\cite{kendall1975thin}.
  By setting $\rho(x_1) = 0$ for a flat surface, which leads to $\psi(a) = \theta$ from~\eqref{eq:psi_NC}, Kendall's result can be immediately recovered from~\eqref{eq:FasFuncofa} which gives the peeling force as
  \begin{linenomath}
  \begin{equation}
    F(a) = \cos(\theta) - 1 + (\left(\cos(\theta) - 1\right)^2 + 2{w})^{1/2},
    \label{eq:force_flat}
  \end{equation}
  \end{linenomath}
  which is a constant for a given nominal peeling angle $\theta$.

  \footnotetext{The symbol $\cp$ denotes the Cartesian product between sets. We use the symbol $F$ to denote both the magnitude of the force acting on the thin film as well as the function defined in~\eqref{eq:rmF_func}.}

  \begin{figure}[t!]
    \centering
    \includegraphics[width=0.6\textwidth]{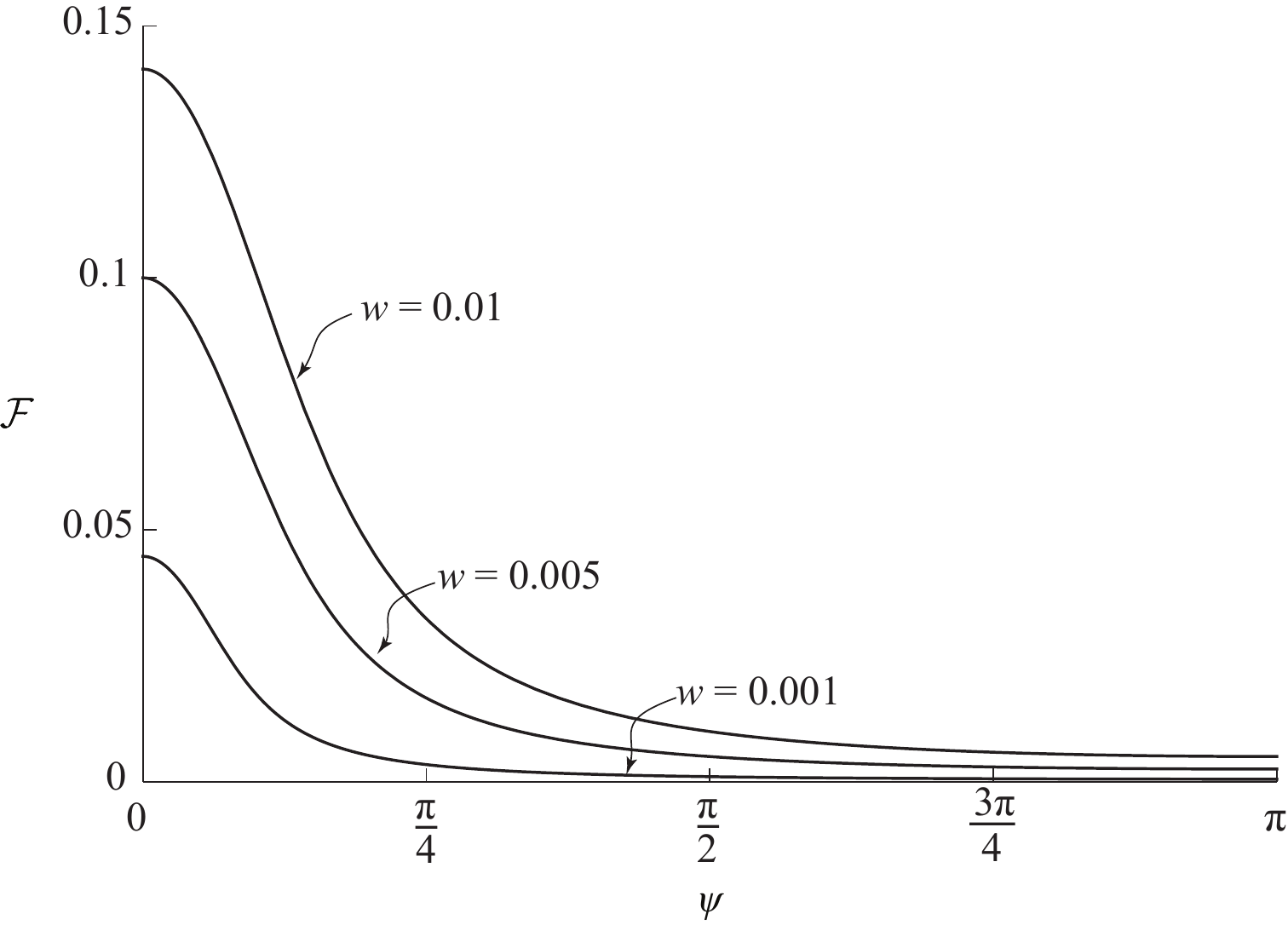}
    \caption{The plot of the graph of $\mathcal{F}$ as $\psi$ increases from 0 to $\pi$ for different ${w}$.}
    \label{fig:F-theta}
  \end{figure}

 For a general peeling process we define the supremum and infimum equilibrium force values, denoted as $F^{+}$ and $F^{-}$, respectively, as
 \begin{linenomath}
  \begin{equation}
      F^{\pm} = \sup/\inf\, \{\mathpzc{F}\pr{\psi(a)}~|~a \in \mathcal{D}\}.
      \label{eq:Fmp}
  \end{equation}
  \end{linenomath}
  It can be shown that the maximum and minimum values of the function $\mathpzc{F}$ are $(2w)^{1/2}$ and $(4+2w)^{1/2}-2$, respectively. This implies that $F^{+}$ is bounded above by $(2w)^{1/2}$ and $F^{-}$ is bounded below by $(4+2w)^{1/2}-2$.

  We denote the maximum and minimum values of the true peeling angle during peeling as $\psi^+$ and $\psi^-$, respectively.
  In the current case of peeling with no contact it follows from the fact that $\mathcal{F}$ is a monotonically decreasing function that
  \begin{linenomath}
  \begin{subequations}
    \begin{align}
      F^{\pm}&=\mathpzc{F}\pr{\psi^{\mp}},\label{eq:Fmp_NC}
      \intertext{where}
      \psi^{\pm} := \theta &+ \tan^{-1}\left(\Dro^{\pm}\right).\label{eq:psi_NC_pm}
    \end{align}
  \end{subequations}
  \end{linenomath} 
  Since $ \mathpzc{F}(\psi)$ is always non-negative it follows from \eqref{eq:Fmp_NC} that $F^{\pm}\ge 0$.

  \begin{theorem}
    If $F \notin [F^{-}, \;F^{+}]$ then there does not exist any equilibrium de-adhered length at that $F$ (e.g., see Figures~\ref{fig:FNonContact}b and e).
    \label{proof:FmpNoEqm}
  \end{theorem}

  \begin{proof}
    Since $\psi^{\pm}$ are the maximum and minimum values of $\psi(a)$, respectively, and from~\eqref{eq:localComp} we have
    \begin{linenomath}
    \begin{equation}
      0\le \psi^{-}\le \psi(a)\le \psi^{+}\le \pi.
      \label{eq:BoundsOnPsi}
    \end{equation}
    \end{linenomath}
    From~\eqref{eq:BoundsOnPsi} and the fact that the cosine function is strictly decreasing in the interval $[0,\pi]$, we find that
    \begin{linenomath}
    \begin{equation}
      -1\le
      \cos\pr{\psi^{+}}
      \le
      \cos\pr{\psi(a)}
      \le
      \cos\pr{\psi^{-}}
      \le 1.
      \label{eq:CosInequalities}
    \end{equation}
    \end{linenomath}
    The inequalities~\eqref{eq:CosInequalities} allow us to express $\cos\pr{\psi(a)}$ as either $\cos\pr{\psi^{-}}-\delta^{-}$ or $\cos\pr{\psi^{+}}+\delta^{+}$; and, though they depend on $a$,  $\delta^{\mp}$,  are always non-negative.

    Let us first consider the case $F>F^{+}$.
    When $F$ is strictly greater than $F^{+}$ it can be represented as $F^{+}+\epsilon^{+}$ for some $\epsilon^{+}>0$.
    Using this representation for $F$ and expressing $\cos\pr{\psi(a)}$ as $\cos\pr{\psi^{-}}-\delta^{-}$ in~\eqref{eq:dPida_NC} we get that
    \begin{linenomath}
    \begin{equation}
      \delta E_{1}(a\, ;\rho, F) =-\frac{1}{2} \left(2 F^+ \left(\delta^-+\epsilon^+\right)+\epsilon^+
      \left(2 \delta^-+\epsilon^{+} + 2\right)\right) \dot{l}(a;\rho)+\cos\pr{\psi^{-}} \epsilon^+ \
      \dot{l}(a;\rho)
      \label{eq:DeltaEOneForFFPlus}.
    \end{equation}
    \end{linenomath}
    Equation~\eqref{eq:l} implies that the derivative of the peeled length $\dot{l}(a;\rho)$ is always positive. For that reason and since $\cos\pr{\psi^{-}}\ge -1$ it follows from~\eqref{eq:DeltaEOneForFFPlus} that when $F>F^{+}$
    \begin{linenomath}
    \begin{equation}
      \delta E_{1}(a\, ;\rho, F)
      \le
      -\frac{1}{2} \left(2 F^+ \left(\delta^-+\epsilon^+\right)+\epsilon^+
      \left(2 \delta^-+\epsilon^+\right)\right) \dot{l}(a;\rho).
      \label{eq:DeltaEOneForFFPlus2}
    \end{equation}
    \end{linenomath}
    It follows from~\eqref{eq:DeltaEOneForFFPlus2} and the facts that $\dot{l}(a;\rho)$ and $\epsilon^+$ are positive and $\delta^{-}$ and $F^+$ are non-negative that when $F>F^+$
    \begin{linenomath}
    \begin{equation}
      \delta E_{1}(a;\rho,F)<0.
      \label{eq:DeltaEOneForFFPlus3}
    \end{equation}
    \end{linenomath}
    Since for equilibrium it is necessary that $\delta E_{1}(a;\rho,F)=0$ it follows from~\eqref{eq:DeltaEOneForFFPlus3} that there can exist no equilibrium de-adhered lengths when $F>F^{+}$.

    Now consider the case $0\le F< F^{-}$.
    When $F<F^{-}$ it can be represented as $F^{-}-\epsilon^{-}$ for some $\epsilon^{-}>0$.
    Representing $F$ this way and expressing $\cos\pr{\psi(a)}$ as $\cos\pr{\psi^{+}}+\delta^{+}$ in~\eqref{eq:dPida_NC} we determine that
    \begin{linenomath}
    \begin{equation}
      \delta E_{1}(a;\rho,F)
      =
      \frac{1}{2} \left(2 F^- \left(\delta ^++\epsilon ^-\right)-\epsilon ^- \left(2 \delta ^++\epsilon ^--2\right)\right)\dot{l}(a;\rho)-\epsilon ^{-}  \cos\pr{\psi^+}\dot{l}(a;\rho).
      \label{eq:DeltaEOneForFFMinus}
    \end{equation}
    \end{linenomath}
    Since $\dot{l}(a;\rho)>0$ and $\cos\pr{\psi^{+}}\le 1$ it follows from~\eqref{eq:DeltaEOneForFFMinus} that when $F<F^{-}$
    \begin{linenomath}
    \begin{equation}
      \delta E_{1}(a;\rho,F)\ge
      \frac{1}{2} \left(2 F^- \left(\delta ^++\epsilon ^-\right)-\epsilon ^- \left(2 \delta ^++\epsilon ^-\right)\right)\dot{l}(a;\rho),
      \notag
    \end{equation}
    \end{linenomath}
    which can be re-arranged to read
    \begin{linenomath}
    \begin{equation}
      \delta E_{1}(a;\rho,F)
      \ge
      \frac{1}{2} \left(2 F \left(\delta ^++\epsilon ^-\right)+\left(\epsilon ^-\right)^2\right)\dot{l}(a;\rho).
      \label{eq:DeltaEOneForFFMinus2}
    \end{equation}
    \end{linenomath}
    Recalling that $\delta^+$ is non-negative and $\epsilon^-$ is positive it follows from~\eqref{eq:DeltaEOneForFFMinus2} that when $0\le F<F^{-}$
    \begin{linenomath}
    \begin{equation}
      \delta E_{1}(a;\rho,F)>0.
      \label{eq:DeltaEOneForFFMinus3}
    \end{equation}
    \end{linenomath}
    For the same reason as before, the inequality~\eqref{eq:DeltaEOneForFFMinus3} implies that when $F$ is less than $F^{-}$ but still non-negative then there cannot exist any de-adhered lengths that are in equilibrium.
  \end{proof}

  The inequality~\eqref{eq:DeltaEOneForFFPlus3} implies that the derivative of the total energy w.r.t $a$ will be negative when $F$ is  greater than $F^{+}$ irrespective of the value of $a$, $w$, or the nature of $\rho$ (see, e.g., Figures~\ref{fig:EnergyCurves}a and c).
  Thus, if $F>F^{+}$ the de-adhered length will grow without bound.
  Realistically, however, the de-adhered length will keep growing until the film completely detaches from the substrate.
  For that reason we call $F^+$ the \textit{peel-off force}.

  The result that~\eqref{eq:DeltaEOneForFFMinus3} holds in the case where $0\le F<F^{-}$ implies that the derivative of the total energy with respect to the de-adhered length  in that case is positive irrespective of any other details in the problem (see, for example, Figures~\ref{fig:EnergyCurves}b and d).
  Therefore, if $0\le F<F^{-}$ and the de-adhered length is initially naught, then the de-adhered length  will not grow or if it is initially  non-zero then it will keep decreasing, i.e., the peel front will keep receding until the entire thin film is adhered to the substrate.
  For this reason, we call $F^{-}$ the \textit{peel-initiation force}.

  %
  \begin{figure}[t!]
    \centering
    \includegraphics[width=\textwidth]{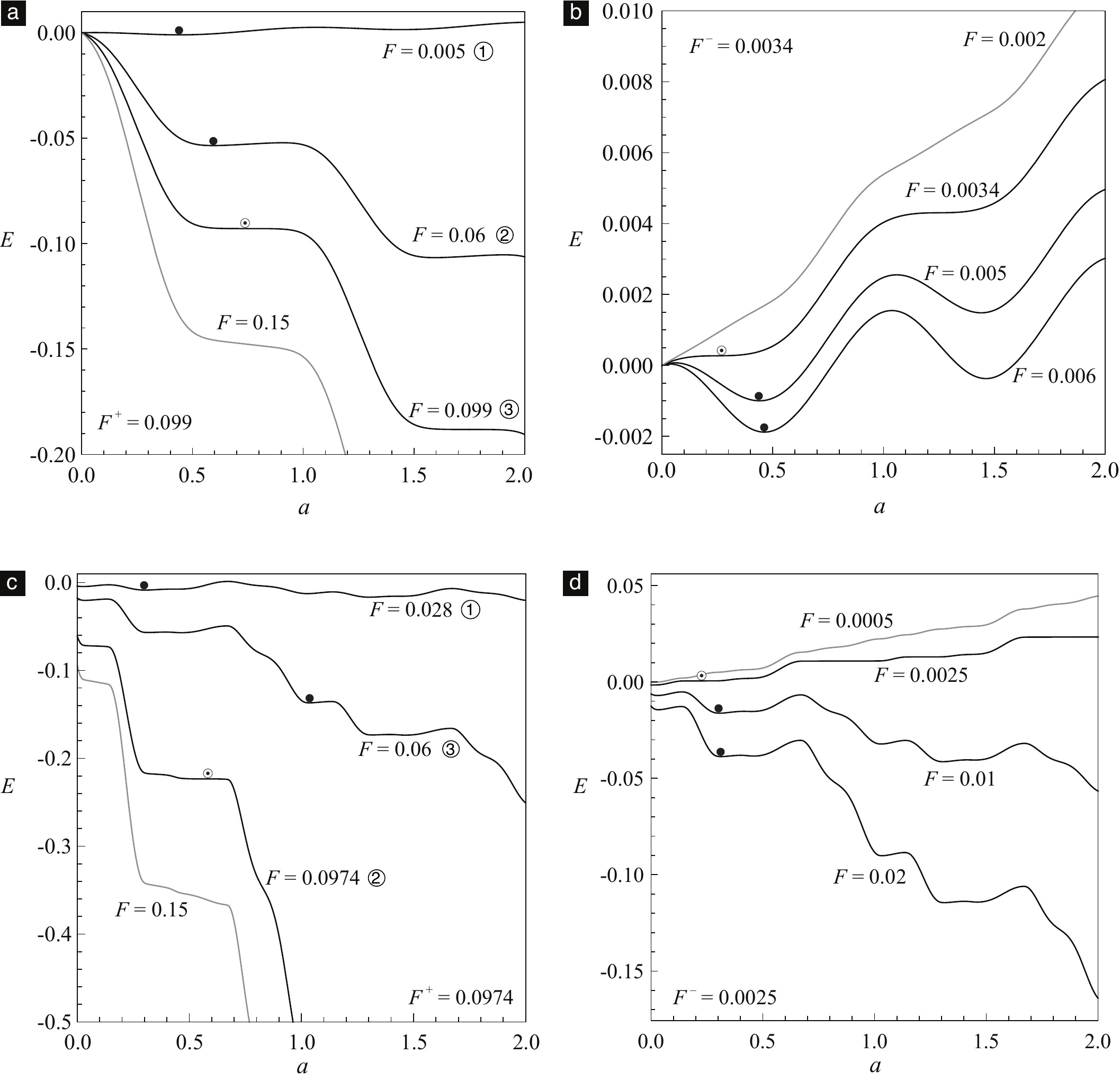}
    \caption{
    The plot of potential energy as a function of ${a}$ for different values of ${F}$ for peeling on the sinusoidal [(a)--(b)] and  complicated [(c)--(d)] surfaces. Note that there is no local minima of the energy when $F > F^+$ in (a) and (c). The labels \textcircled{1}--\textcircled{3} correspond to those marked in Figures~\ref{fig:FNonContact}b and d. There is also no local minima of the energy when $0 \le F < F^-$ in (b) and (d).
    }
    \label{fig:EnergyCurves}
  \end{figure}

  \subsubsection{Stability of equilibrium state}
  \label{sec:StaPeelingNoContact}
  We study the stability  of the local equilibria  by examining the sign of the second variation of the total potential energy. 
  Specifically, a configuration with the de-adhered length $a$ is  a stable equilibrium state \textit{iff} $a$ belongs to the set
  \begin{linenomath}
  \begin{equation}
      \Dst(F,\rho):=\left\{a\in \mathcal{D}^{\circ}(F,\rho)~|~\delta E_{2}(a\,;\rho,F)> 0\right\}.
      \label{eq:StableDomain}
  \end{equation}
  \end{linenomath}
  It is a neutral equilibrium state \textit{iff} $a$ belongs to the set
  \begin{linenomath}
  \begin{equation}
      \mathcal{D}^{\odot}(F,\rho):=\left\{a\in \mathcal{D}^{\circ}(F,\rho)~|~\delta E_{2}(a\,;\rho,F)= 0\right\},
      \label{eq:NeutralDomain}
  \end{equation}
  \end{linenomath}
  and is an unstable equilibrium state if $a$ belongs to the set $\mathcal{D}^{\otimes}(F,\rho)$, which is a set of all de-adhered lengths that belong to $\mathcal{D}^{\circ}(F,\rho)$ but not to $\Dst(F,\rho)$   or $\mathcal{D}^{\odot}(F,\rho)$. 

\paragraph{Stability and surface curvature}
Suppose $a$ is an equilibrium de-adhered length at the force $F$. 
Then it is necessary that $F$ and $a$ satisfy the  equation $F=F(a)$, where $F(\cdot)$ on the right hand side  is the function defined in~\eqref{eq:FasFuncofa}. 
Upon substituting $F$ in~\eqref{eq:d2Pida2_NC} with $F(a)$, and then simplifying the resulting equation,
we get that
\begin{linenomath}
\begin{equation}
  \delta E_{2}(a\, ;\rho, F(a))
  =
  -\frac{F(a)}{2} k({a}) \Dl(a;\rho)^2 \sin\pr{\psi(a)},
  \label{eq:d2Pida2_NC1}
\end{equation}
\end{linenomath}
where
\begin{linenomath}
\begin{equation}
    k({a}) = \DDro({a})/\dot{l}(a;\rho)^3
    \label{eq:curvature}
\end{equation}
\end{linenomath}
is the signed curvature of the graph of $\rho$. 
The mean curvature of the substrate's surface at the point whose coordinates w.r.t to $\physe_1$ and $\physe_2$ are $a$ and $\rho(a)$, respectively, equals $k(a)/(2\lambda)$. 
Therefore, we will often refer to $k(a)$ as the substrate's surface curvature. 

If $w=0$, then from~\eqref{eq:FasFuncofa} we have that $F(a)=0$. Consequently, from~\eqref{eq:d2Pida2_NC1},  $\delta E_2=0$ for all $a$. 
Therefore, when $w=0$ all states are neutral-equilibrium states. 
In the following section we take that $w$ is positive (recall that $w\ge 0$). 

If $k(a)$ vanishes then it follows from~\eqref{eq:d2Pida2_NC1} that $a$ belongs to $\Dnt$, i.e., that the corresponding state is a neutral equilibrium state. 

It follows from the definitions of $a^{\pm}$, $\Dro^{\pm}$, and $k$ and the smoothness of $\rho$ that $k\pr{a}=0$ \textit{iff} $a=a^{\pm}$. 
So, if $k(a)$ does not vanish then $a$ is different from $a^{\pm}$, which implies from the definitions of  $\Dro^{\pm}$, $\psi$, and $\psi^{\pm}$ that $\psi(a)$ is different from $\psi^{\pm}$. 
This last deduction in conjunction with \eqref{eq:BoundsOnPsi} implies that when $k(a)$ is not naught, $\sin(\psi(a))$ is positive. 
Hence, it follows from~\eqref{eq:d2Pida2_NC1} that the configuration is stable (resp. unstable)
when $k(a)$ is negative (resp. positive). 
These results, which connect an equilibrium state's stability to the surface curvature, are illustrated by  Figures~\ref{fig:FNonContact}b--c for a sinusoidal surface and by Figures~\ref{fig:FNonContact}e--f for a complicated surface. 

\paragraph{Positive, negative, and zero values of the function $\DF(\cdot)$ imply stable, unstable, and neutral equilibria, respectively }

  Again, let $a$ be an equilibrium de-adhered length at the force $F$. 
From~\eqref{eq:FasFuncofa} we have that
  \begin{linenomath}
  \begin{equation}
    \dot{F}(a)
    =
    -k(a)
    \Dl(a;\rho)
    \sin(\psi(a))
    \left( 1 - \frac{1-\cos(\psi(a))}{(2{w}+(1-\cos(\psi(a)))^2)^{1/2}}\right).
    \label{eq:dFda}
  \end{equation}
  \end{linenomath}
Recall that we deduced that  when $w=0$ all  configurations are  neutral equilibrium configurations. 
Therefore, in the   following two paragraphs we take that $w>0$. 

Say $\dot{F}(a)$ vanishes. 
It can be checked using~\eqref{eq:Dl} that  $\Dl(a;\rho)$ is always positive, and
since we have assumed that $w>0$ it can be shown that the expression within the large parenthesis on the right hand side of~\eqref{eq:dFda} is always positive. 
Therefore, if $\dot{F}(a)$ vanishes then we have the following three cases from~\eqref{eq:dFda}:\textit{ (i)} the factor $k(a)$ vanishes, \textit{(ii)} the factor $\sin(\psi(a))$ vanishes, \textit{(iii)} both these factors vanish. 
The factor $k(a)$ vanishes in both case \textit{(i)} and \textit{(iii)}. 
Let us focus on case \textit{(ii)}. 
If  $\sin(\psi(a))$ vanishes
then we know from~\eqref{eq:BoundsOnPsi} that  $a=a^{\pm}$, which then implies, based on the discussion following~\eqref{eq:d2Pida2_NC1}, that $k(a)$ has to also vanish.  
then we know from~\eqref{eq:BoundsOnPsi} that  $a=a^{\pm}$, which then implies, based on the discussion contained in the third paragraph following the one containing~\eqref{eq:d2Pida2_NC1} that, $k(a)$ has to also vanish.  
Thus, $k(a)$ vanishes in all three cases. 
That is,  if $\dot{F}(a)$ is equal to zero  then $k(a)$ is also equal to zero. 
This last deduction in light of the results presented in \textit{Stability and surface curvature} implies that if  $\dot{F}(a)$ vanishes then the configuration corresponding to $a$ is a neutral-equilibrium configuration. 


Now say that $\dot{F}(a)$ is positive (resp. negative). 
As previously stated, since we have assumed that $w>0$ the expression within the large paranthesis on the right hand side of~\eqref{eq:dFda} is always positive. 
The factor $\sin(\psi(a))$ is positive since we can show using~\eqref{eq:BoundsOnPsi} that it is always non-negative and if it were to vanish then that would contradict the assumption that  $\dot{F}(a)$ is non-zero. 
This, in conjuction with~\eqref{eq:dFda}, imply that $k(a)$ is negative (resp. positive) whenever  $\dot{F}(a)$ is positive (resp. negative). 
In light of the results presented in  \textit{Stability and surface curvature}, this last deduction implies that when $\dot{F}(a)$ is positive (resp. negative) then the corresponding equilibrium configuration is stable  (resp. unstable). 

These results are illustrated in Figures~\ref{fig:FNonContact}b and e using the example surface profiles shown in Figures~\ref{fig:FNonContact}c and f. 

\subsection{Peeling process that might involve contact}
\label{sec:PeelingWithContact}

\begin{figure}[t!]
    \centering
    \includegraphics[width=.8\textwidth]{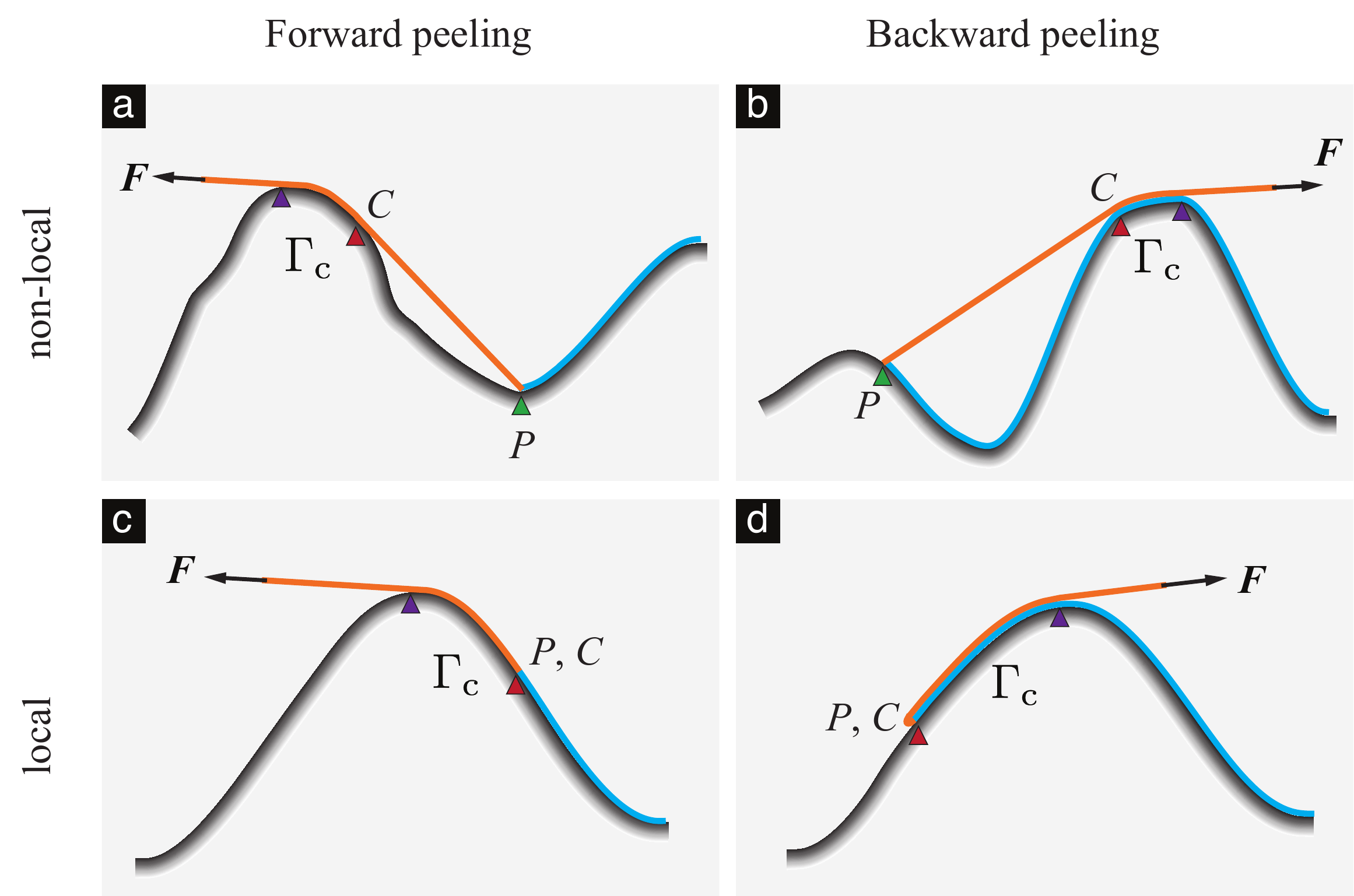}
    \caption{A schematic of peeling with contact. (a) Forward peeling, non-local contact. (b) Backward peeling, non-local contact. (c) Forward peeling, local contact. (d) Backward peeling, local contact.}
    \label{fig:CasesofContactConfiguration}
\end{figure}

If $\theta$ is kept fixed during the peeling process and that constant $\theta$ violates the global compatibility condition~\eqref{eq:GlobalComp}, then there will exist some configurations during the peeling process that will involve contact\footnotemark. 
We provide a procedure for determining whether or not a given configuration involves contact in \S\ref{sec:ProcedureCatofKappa}. 
If the configuration $\boldsymbol{\kappa}$ does not involve contact, the results needed to compute the force $F$ so that  $(\conf, F)$ is an equilibrium state and the results needed to determine the stability of  $(\conf, F)$ are given in \S\ref{sec:PeelingNoContact}. 

\footnotetext{
This, of course, does not mean that in such a peeling process all configurations will involve contact. %
That is, there can exist configurations that involve no contact during parts of the peeling process (see, e.g., subfigure (b) in Figure~\ref{fig:SinFContact}).
}

When $\conf$ involves contact the primitive conditions that determines whether or not a state $(\conf, F)$ is in equilibrium remain the same as before.  
Specifically, even when a configuration $\boldsymbol{\kappa}$ involves contact, the state $(\conf, F)$ is an equilibrium state \textit{iff} the de-adhered length $a$ in $\conf$ belongs to the set $\mathcal{D}^{\circ}(F,\rho)$, which is defined in \eqref{eq:EqmDomain}. 
Similarly, it qualifies as a stable, neutral, or unstable state depending on whether the de-adhered length $a$  belongs to the sets $\mathcal{D}^{\bullet}(F,\rho)$, $\mathcal{D}^{\odot}(F,\rho)$, or $\mathcal{D}^{\otimes}(F,\rho)$, respectively. These sets are defined in \S\ref{sec:StaPeelingNoContact}. 

What makes the analysis of the case involving contact more challenging is the calculation of the functions $\delta E_1$ and $\delta E_2$, which are needed for the construction of the sets $\mathcal{D}^{\circ}(F,\rho)$,  $\mathcal{D}^{\bullet}(F,\rho)$, etc. These functions are defined in \eqref{eq:dPi_NC} in terms of the asymptotic expansion of $\delta E$ as $\delta a \to 0 $. The calculation of $\delta E$'s asymptotic expansion is challenging due to the presence of the term $- \boldsymbol{F}\cdot \delta \boldsymbol{u}$ in~\eqref{eq:dE}. In the case involving contact the term $ -\boldsymbol{F}\cdot \delta \boldsymbol{u}$ simplifies to $ -F \delta u$, where $\delta u$ is $\delta \boldsymbol{u}$'s magnitude. This is due to the fact that $\delta \boldsymbol{u}$  is always in the same direction as $\boldsymbol{F}$ (see Figures~\ref{fig:ForwardPeeling} and \ref{fig:BackwardPeeling}) when a configuration involves contact. This remains true irrespective of whether the contact region consists of a single contact patch (e.g., see Figure~\ref{fig:CasesofContactConfiguration}) or several contact patches (e.g., see subfigure (a) in either Figure~\ref{fig:ForwardPeeling} or \ref{fig:BackwardPeeling}). Thus, calculation of $\delta E$'s asymptotic expansion requires the calculations of $\delta u$'s asymptotic expansion, which in this current case is non-trivial. To elaborate, in the case not involving contact, it is straightforward to determine the asymptotic expansion of $\delta u$, e.g. through the use of~\eqref{eq:du_NC} and~\eqref{eq:deltalAsymExpan}. While now in the case involving contact this exercise is relatively more difficult.

We could not obtain a general, closed-form expression for the asymptotic expansion of $\delta u$ when the configuration involved contact. However, in~\S\ref{sec:EnergyVariation_C} we present a family of four analytical, but not closed-form,  expressions for calculating $\delta u$'s asymptotic expansion,  and from that $\delta E_1$ and $\delta E_2$, that apply to special categories of contact configurations. We describe these four categories, to whom we henceforth refer to as \textit{C.1}, \textit{C.2}, etc.,   shortly in~\S\ref{subsubsec:FourCategories}, but we note here that it will follow from their definitions that any  contact configuration can be uniquely placed into in one of them.

From the family of $\delta E_1$ functions given in \S\ref{sec:EnergyVariation_C}, which apply to different categories of contact configurations, we found that,  interestingly,  irrespective of which category a contact configuration $\conf$ belongs to the force $F$ needed to make the state $(\conf, F)$  an equilibrium state is always   $\mathcal{F}(\psi(a))$, where $\mathcal{F}$ is defined in~\eqref{eq:F_func}. However, it follows from the family of $\delta E_2$ functions given in \S\ref{sec:EnergyVariation_C} that $\conf$'s category is still relevant for determining the nature of $(\conf, F(a))$'s stability.

In summary, our method for simulating a peeling process in which $\theta$ violates the  global compatibility condition at some stage of the peeling process is as follows. 
Let the peeling experiment be defined by prescribing the sequence $(a_i, \theta_i)$, where $i=1$, $2,$ etc., and the symbols $a_i$ and $\theta_i$ are the de-adhered length and the nominal peel angle, respectively, in the $i^{\rm th}$ step of the experiment. 
We compute $\psi(a_i)$, the true peeling angle for the  $i^{\rm th}$ step, using Algorithm~\ref{algo:TruePeelingAngle}. 
We place the configuration $\conf_i$  into one of the four categories, \textit{C.1--4}, using  $\theta_i$ and $\psi(a_i)$ (see \S\ref{sec:ProcedureCatofKappa} for details). 
We then compute the  force  $F_i$  such that the state $\mathcal{S}_i=(\conf_i,F_i)$ becomes an equilibrium state as $\mathcal{F}(\psi(a_i))$. 
We determine the nature of $\mathcal{S}_i$'s stability  by computing $\delta E_2(a_i;\rho, F_i)$ and constructing the sets $\Dst$, $\Dnt$, etc. 
When the contact configuration $\conf$ belongs to either \textit{C.1} or \textit{C.2} the Algorithm~\ref{algo:TruePeelingAngle} also provides the value of the parameter $\ell$, which is the distance between the peeling and contact fronts. So, when $\conf$ belongs to either \textit{C.1} or \textit{C.2} we use that value in conjunction with~\eqref{eq:d2Pida2_C12} to compute $\delta E_2(a_i;\rho, F_i)$.
When  $\conf$ belongs to \textit{C.3} or \textit{C.4} we compute the value of $\delta E_2$ using~\eqref{eq:d2Pida2_C3} or \eqref{eq:d2Pida2_C4}, respectively. 

We demonstrate our method by using the  same  two example surfaces that we previously considered in \S\ref{sec:PeelingNoContact}. The schematics of these simple and complicated wavy surfaces are shown, e.g., in Figures~\ref{fig:FNonContact}a and d, respectively.

On each surface we simulated two (virtual) peeling experiments. In the first experiment---forward peeling (defined in~\S\ref{subsubsec:FourCategories})---the peeling angle was kept fixed at a value of $\pi/3$ through out the experiment, while in the second one---backward peeling (defined in~\S\ref{subsubsec:FourCategories})---it was kept fixed at $3\pi/4$. The constant nominal peeling angles we chose, namely $\pi/3$ and $3\pi/4$, violated the global compatibility condition on both our example surfaces. Therefore, the results from \S\ref{sec:PeelingNoContact} cannot be used to simulate these experiments,  for  instance to generate the set
\begin{linenomath}
\begin{equation}
\bigcup_{F\in [F^{-},F^{+}]}\mathcal{D}^{\circ}(F,\rho)\cp \{F\}
\label{def:EquilibriumStates}
\end{equation}
\end{linenomath}
for these experiments. Note that a typical point in~\eqref{def:EquilibriumStates} represents an equilibrium state, with its abscissa denoting the state's de-adhered length and its ordinate the state's force. Therefore, we applied our method, which we introduced earlier in this section, to the sequence $(a_i, \theta_0)$ ($\theta_0=\pi/3,~3\pi/4$)  and computed the  sequence $(F_i)$, where $F_i$ is the equilibrium force corresponding to $a_i$; And constructed (a subset of)~\eqref{def:EquilibriumStates}, alternately, as $\left\{(a_i,F_i)\right\}$. The sets~\eqref{def:EquilibriumStates} that we generated this way for the forward and backward peeling cases are shown  in subfigures (b) and (e), respectively, of Figure~\ref{fig:SinFContact} for the simple wavy surface and in Figure~\ref{fig:ZigFContact} for the complicated wavy surface.

Note that our method also determines the stability of an equilibrium state and informs us whether or not that state involves contact. In the subfigures (b) and (e) of Figures~\ref{fig:SinFContact} and~\ref{fig:ZigFContact} the stable equilibrium states are denoted using solid/filled symbols while unstable states are denoted using hollow/unfilled symbols. In the subfigures we identify the states that involve contact by placing them over a yellow background. As can be noted from the subfigure, a yellow region is preceded  and followed by white regions, and \textit{vice versa}. Thus, when the global compatibility condition is violated a sequence of configurations involving contact can be followed by a sequence  of configurations not involving contact, and so on.

Finally, our method also provides the true peeling angle sequences $\pr{\psi(a_i)}$ in the experiments. These are shown in subfigures (a) and (b) of Figures~\ref{fig:SinFContact} and~\ref{fig:ZigFContact}.

A few representative configurations from the peeling experiments are explicitly sketched in subfigures (c) and (f) of Figures~\ref{fig:SinFContact} and~\ref{fig:ZigFContact}.

\subsubsection{The four categories of contact configurations}
\label{subsubsec:FourCategories}
A configuration involving contact can be placed into one of the following four categories. 
\begin{enumerate}[label=\textit{C.\arabic*}]
\item  Forward-peeling, non-local contact (Figure~\ref{fig:CasesofContactConfiguration}a),
\item  Backward-peeling, non-local contact (Figure~\ref{fig:CasesofContactConfiguration}b),
\item  Forward-peeling, local contact (Figure~\ref{fig:CasesofContactConfiguration}c), and
\item  Backward-peeling, local contact (Figure~\ref{fig:CasesofContactConfiguration}d).
\end{enumerate} 
We call a configuration a \textit{forward-peeling}  configuration if  the  $\theta$ in it is less than $\pi/2$, and a  \textit{backward-peeling} configuration otherwise. 
Roughly speaking, we consider configurations of the type shown in Figures~\ref{fig:CasesofContactConfiguration}a and b as those that involve non-local contact, and configurations of the type shown in  Figures~\ref{fig:CasesofContactConfiguration}c and d as those that involve local contact.  
We define local and non-local contact precisely  by introducing the notions of \textit{contact region} and \textit{contact front}, which we discuss next. 

We define the contact region  corresponding to the deformed configuration $\boldsymbol{\kappa}$ as $\Gamma_{c}=\{x_1 \in \mathcal{D}~|~\boldsymbol{x}(x_1)\in \partial\mathbb{S}\}$, where 
\begin{linenomath}
\begin{equation}
    \partial \mathbb{S}
    =\left\{x_1\physe_1+x_2\physe_2+x_3\physe_3 \in \mathbb{E}~|~(x_1,x_2,x_3)\in\mathbb{R}^3~\text{and}~x_2 = \rho(x_1) \right\}
    \label{eq:SubstrateSurfaceGeometry}
\end{equation}
\end{linenomath}
is the substrate's surface (cf.~\eqref{eq:SubstrateGeometry}) and $\boldsymbol{x}$ is defined in \eqref{eq:DeformationMappingFromManifoldToCurrent}. Let $c$ be the point in $\Gamma_c$ that is closest to $\Gamma_a$ in $\mathcal{D}$'s topology; recall here that $\Gamma_a$ is the adhered region corresponding to the configuration $\boldsymbol{\kappa}$. We define the contact front as $C=\{O+\boldsymbol{x}(c)+x_3\physe_3\in \mathcal{E}~|~\lvert x_3\rvert \le b/2  \}$.

We say that a deformed configuration $\conf$ involves local contact \textit{iff} $C=P$ and involves non-local contact  otherwise.

\subsubsection{Procedure for determining a configuration's type---contact or non-contact---and a contact configuration's category}
\label{sec:ProcedureCatofKappa}

\begin{algorithm}[ht!]
\caption{Procedure for determining the category of a  configuration}\label{algo:TruePeelingAngle}
\begin{algorithmic}[1]
\Procedure {Generate } {true peeling angle $\psi(a)$ and distance $\ell$}

\State{\textbf{Input:} Substrate's surface profile $\rho$, nominal peeling angle $\theta$, and de-adhered length $a$}

\State{Let $\Delta \rho(x_1) := \rho(a)-\tan\pr{\theta}\pr{x_1-a}-\rho\pr{x_1}$, where $x_1\in \mathbb{R}$.}

\If{$\Delta \rho(x_1) > 0$ for all $x_1$ satisfying $\text{sgn}\pr{(\theta-\pi/2)(x_1-a)} > 0$\footnotemark}
    \Comment{No contact}
    \State{
        \begin{equation*}
            \psi(a) = \theta + \tan^{-1}(\Dro(a)).
        \end{equation*} \Comment{c.f. \eqref{eq:psi_NC}}
    }

\Else
    \If{$k(a) \le 0 $ and $\Dro(a) \text{sgn}(\theta-\pi/2) > 0$} \Comment{Local type contact}
    \State{
        \begin{equation*}
            \psi(a) = \begin{cases}
                        0, & \theta < \pi/2, \\
                        \pi, & \theta \ge \pi/2.
                      \end{cases}
        \end{equation*}
    }
    \Else \Comment{Non-local type contact}
    \State{Find the abscissa of contact front $c$ such that
        \begin{equation*}
            \frac{\rho(a)-\rho(c)}{a-c} = \Dro(c) \quad \text{for} \quad \mathrm{sgn}\pr{(\theta-\pi/2)(c-a)} > 0.
        \end{equation*}
    }
    \State{Then
        \begin{equation*}
            \psi(a) = \begin{cases}
                        \tan^{-1}\pr{\Dro(a)} - \tan^{-1}\pr{\Dro(c)}, & \theta < \pi/2, \text{ forward peeling,} \\
                        \pi + \tan^{-1}\pr{\Dro(a)} - \tan^{-1}\pr{\Dro(c)}, & \theta \ge \pi/2, \text{ backward peeling,}
                      \end{cases}
        \end{equation*}
    }
    \State{and
        \begin{equation*}
          \ell = \pr{(a-c)^2 + \pr{\rho(a)-\rho(c)}^2}^{1/2}.
        \end{equation*}
    }
    \EndIf
\EndIf
\EndProcedure
\end{algorithmic}
\end{algorithm}

\footnotetext{We define the $\text{sgn}(\cdot)$ function  as:
\begin{equation*}
    \text{sgn}(x)= \begin{cases}
                        +1, & x \ge 0, \\
                        -1, & x < 0.
                    \end{cases}
\end{equation*}
}

When $\theta$ satisfies the global compatibility condition then we know that the configuration $\conf$ will be of non-contact type. When the global compatibility condition is violated, as we describe in the next few paragraphs, determining whether or not $\conf$ involves contact, i.e., determining its type, and if it involves contact then determining the contact category that $\conf$ belongs to essentially comes down to determining the true peeling angle $\psi(a)$.

The true peeling angle is defined in \S\ref{subsubsec:ContactConditions}. When $\theta$ satisfies the global compatibility condition  $\psi(a)$ is given by~\eqref{eq:psi_NC}. When  $\theta$ violates the the global compatibility condition it can be  computed using the numerical procedure that we present in Algorithm~\ref{algo:TruePeelingAngle}.

Given a configuration $\conf$, if the true peeling angle $\psi(a)$ in it is different from~~\eqref{eq:psi_NC} then $\conf$ is a contact type configuration. Otherwise, it is a non-contact type configuration.

For placing a contact type configuration $\conf$ into one of the four categories described in \S\ref{subsubsec:FourCategories} it is sufficient to  know whether $\conf$ is a forward or backward peeling configuration and whether the contact in it is of the local or the non-local type.

The configuration $\conf$ is a forward peeling configuration if $\theta<\pi/2$, and a backward peeling configuration otherwise. 
A contact configuration $\conf$ involves non-local or local contact depending on whether the true peeling angle in it, $\psi(a)$, lies in the interior or on the boundary of the set $[0,\pi]$.

  \subsubsection{Asymptotic expansion of \texorpdfstring{$\delta u$}{} and the functions \texorpdfstring{$\delta E_1$}{} and  \texorpdfstring{$\delta E_2$}{} for the different categories of contact configurations}
  \label{sec:EnergyVariation_C}

  \paragraph*{Categories \textit{C.1--2}}

  In~\S\ref{sec:Appen_du} we show that for these categories
  \begin{linenomath}
  \begin{equation}
  \begin{aligned}
    \delta {u} =
    & \left(1+\varepsilon - \cos(\psi(a)) \right) \dot{l}({a};\rho) \delta{a} \\
    &+ \frac{1}{2}\left( (1+\varepsilon) \ddot{l}({a};\rho) - \frac{\dot{l}({a};\rho)^2 }{\ell} \sin^2(\psi(a)) + \frac{\pr{\sin(\psi(a))-\cos(\psi(a))\Dro(a)}\DDro(a)}{\Dl(a;\rho)} \right) (\delta {a})^2 + o((\delta {a})^2),
    \label{eq:du_C}
  \end{aligned}
  \end{equation}
  \end{linenomath}
  where $\ell$ is the distance between $P$ and $C$. Recall that $P$ and $C$ denote the peeling and  contact front, respectively (e.g., see Figure~\ref{fig:CasesofContactConfiguration}). We introduced these notions in \S\ref{sec:geometry} and \S\ref{subsubsec:FourCategories}.
  Substituting the $\delta u$ appearing in~\eqref{eq:dE} with the expression appearing on the right hand side of~\eqref{eq:du_C} and then comparing the resulting equation with~\eqref{eq:dPi_NC} we get that
   \begin{linenomath}
    \begin{subequations}
      \begin{align}
        \delta E_{1}(a\,;\rho, F) &:= \left( -\frac{1}{2}{F}^2 + {F}(\cos(\psi(a))-1) + {w} \right)\dot{l}({a};\rho), \label{eq:dPida_C} \\
        \delta E_{2}(a\,;\rho, F) &:= - \frac{1}{2} \left( \left(\frac{{F}^2}{2}+{F}-{w} \right) \ddot{l}({a};\rho) - F \left( \frac{\dot{l}(a;\rho)^2}{\ell} \sin^2(\psi(a)) - \frac{\pr{\sin(\psi(a))-\cos(\psi(a))\Dro(a)}\DDro(a)}{\Dl(a;\rho)} \right) \right) \label{eq:d2Pida2_C12}.
      \end{align}
    \end{subequations}
  \end{linenomath}

\paragraph*{Categories \textit{C.3}}
As can be noted from Figure~\ref{fig:ForwardPeeling}c for this category, $\delta u = \varepsilon\delta l$. With~\eqref{eq:deltalAsymExpan} we have
\begin{linenomath}
\begin{equation}
  \delta{u} = \varepsilon \pr{\dot{l}({a};\rho) \delta{a} + \frac{1}{2} \ddot{l}({a};\rho) (\delta {a})^2 + o\pr{(\delta {a})^2}}.
  \label{du_C3}
\end{equation}
\end{linenomath}
 As before, substituting  $\delta u$ from~\eqref{du_C3} into~\eqref{eq:dE} and comparing the resulting equation with~\eqref{eq:dPi_NC} we get that
\begin{linenomath}
  \begin{subequations}
    \begin{align}
      \delta E_{1}(a\, ;\rho, F) &:= -\left( \frac{1}{2}{F}^2 - {w} \right)\dot{l}({a};\rho), \label{eq:dPida_C_iii_f} \\
      \delta E_{2}(a\, ;\rho, F) &:= - \frac{1}{2} \left(\frac{{F}^2}{2}-{w} \right) \ddot{l}({a};\rho). \label{eq:d2Pida2_C3}
    \end{align}
  \end{subequations}
\end{linenomath}

\paragraph*{Categories \textit{C.4}}
As can be noted from Figure~\ref{fig:BackwardPeeling}c for this category, $\delta u = (1+\varepsilon)\delta l$. With~\eqref{eq:deltalAsymExpan} we have
\begin{linenomath}
\begin{equation}
  \delta{u} = (1+\varepsilon) \pr{\dot{l}({a};\rho) \delta{a} + \frac{1}{2} \ddot{l}({a};\rho) (\delta {a})^2 + o\pr{(\delta {a})^2}}.
  \label{du_C4}
\end{equation}
\end{linenomath}
It follows from~\eqref{du_C4}, \eqref{eq:dE}, and \eqref{eq:dPi_NC} that
\begin{linenomath}
\begin{subequations}
\begin{align}
    \delta E_{1}(a\, ;\rho, F) &:= -\left( \frac{1}{2}{F}^2 +2F - {w} \right)\dot{l}({a};\rho), \label{eq:dPida_C_iii_b} \\
    \delta E_{2}(a\, ;\rho, F) &:= - \frac{1}{2} \left(\frac{{F}^2}{2}+2F-{w} \right) \ddot{l}({a};\rho). \label{eq:d2Pida2_C4}
\end{align}
\end{subequations}
\end{linenomath}
\begin{figure}[t!]
\centering
\includegraphics[width=\textwidth]{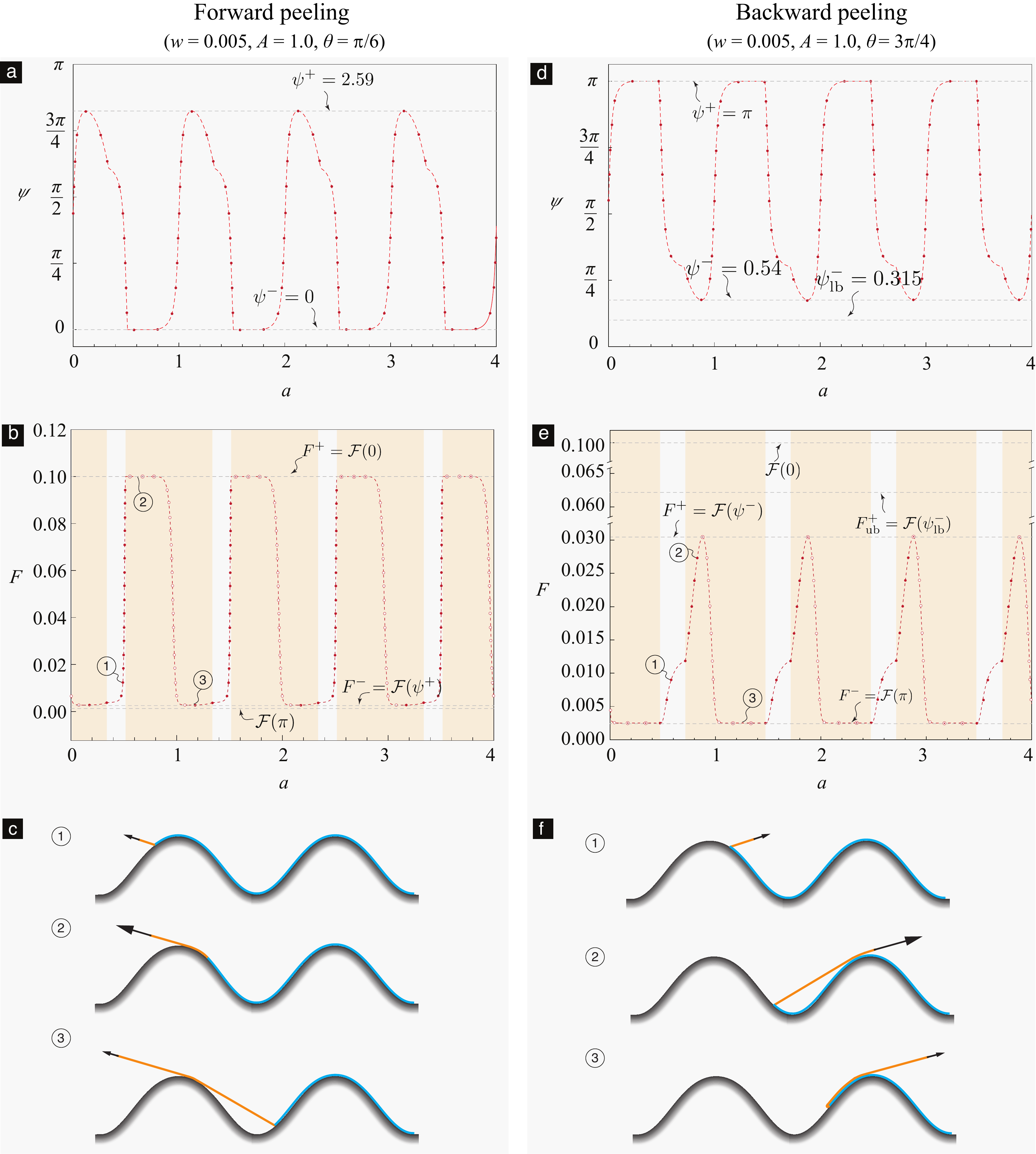}
\caption{Forward peeling (a--c) and backward peeling (d--f) on a sinusoidal surface. The function $\varrho$ of surface profile is the same as the one considered in Figure~\ref{fig:FNonContact}a. (a) and (d) show the numerically calculated true peeling angle $\psi$, (b) and (e) show the ${F}$--${a}$ plot, and (c) and (f) show representative peeling configurations corresponding to the equilibrium states marked in (b) and (e), respectively. In (b) and (e), the stable, neutral, and unstable equilibrium state are marked with a solid dot, circled dot, and circle, respectively. The yellow regions indicate the occurrence of contact during peeling, while white regions indicate no contact.
For the forward (resp. backward) peeling, the peel-off force $F^+ = 0.1$ (resp. $F^+=0.03$), which is much greater than the peel-off force, 0.033 (resp. 0.003), for peeling on a flat surface with the same nominal peeling angle.
}
\label{fig:SinFContact}
\end{figure}

\begin{figure}[ht!]
\centering
\includegraphics[width=\textwidth]{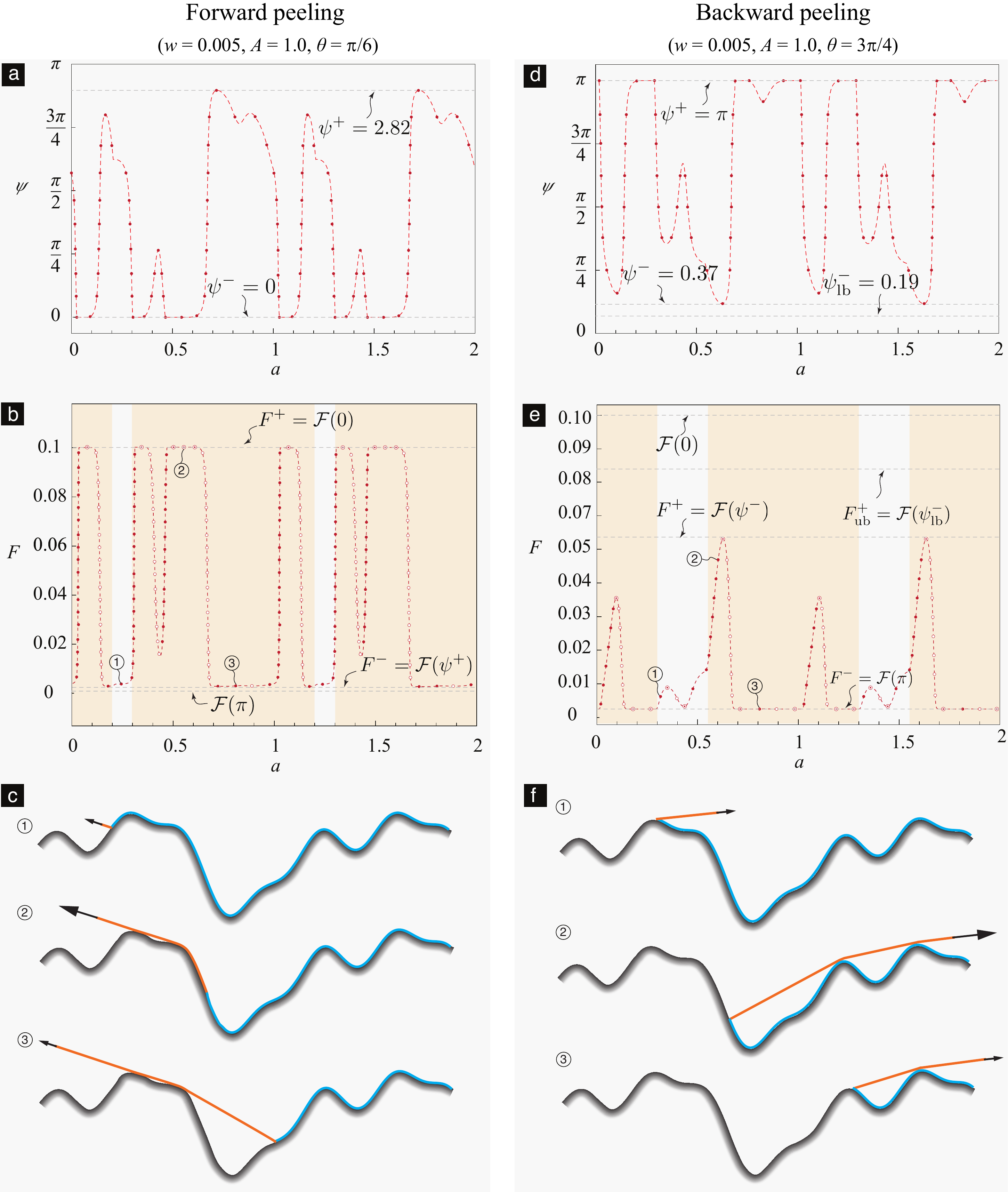}
\caption{Forward peeling (a--c) and backward peeling (d--f) on a complicated surface. The function $\varrho$ of surface profile is the same as the one considered in Figure~\ref{fig:FNonContact}d. (a) and (d) show the numerically calculated true peeling angle $\psi$, (b) and (e) show the ${F}$--${a}$ plot, and (c) and (f) show representative peeling configurations corresponding to the equilibrium states marked in (b) and (e), respectively. The stable, neutral, and unstable equilibrium states are marked with a solid dot, circled dot, and circle, respectively. The yellow regions indicate the occurrence of contact during peeling, while white regions indicate no contact.
For the forward (resp. backward) peeling, the peel-off force $F^+ = 0.1$ (resp. $F^+=0.053$), which is much greater than the peel-off force, 0.033 (resp. 0.003), for peeling on a flat surface with the same nominal peeling angle.
}
\label{fig:ZigFContact}
\end{figure}

  \subsubsection{Remarks on peeling with contact}
  \label{sec:ContactRemarks}
    The results for peeling involving contact shown in Figures~\ref{fig:SinFContact} and~\ref{fig:ZigFContact} only apply to the sample surfaces shown in Figures~\ref{fig:FNonContact}a and d, respectively. 
    That is, we do not have a general, closed-form, analytical theory for the case in which at least some configurations involve contact. 
    However, we can make the following general, interesting, remarks with regard to the case involving contact. 

Theorem~\ref{proof:FmpNoEqm} also holds when the global compatibility condition~\eqref{eq:GlobalComp} is violated.

\begin{theorem}
  During forward peeling when the global compatibility condition~\eqref{eq:GlobalComp} is violated, specifically when $\theta < -\tan^{-1}\pr{\Drom}$, the peel-off force achieves its upper bound, which is $(2w)^{1/2}$. 
  \label{proof:forward_ub}
\end{theorem}

\begin{proof}

Consider the configuration in which the de-adhered length  $a=a^{-}$.
The length $a^{-}$ is defined in~\eqref{eq:aPlusMinus}.
Using Algorithm~\ref{algo:TruePeelingAngle} we  show that the true peeling angle in this configuration, $\psi(a^-)$, is naught.

We start by showing that the $a^{-}$ configuration involves contact. Let us assume that the $a^{-}$ configuration does not involve contact. It then follows from Algorithm~\ref{algo:TruePeelingAngle}, line 4 that $\Delta\rho(x_1)>0$ for all $x<a^{-}$. Recalling $\Delta\rho$'s definition this last implication can be written more explicitly as
\begin{linenomath}
\begin{equation}
    \frac{\rho(a)-\rho(x_1)}{a-x_1} > - \tan(\theta),
\label{eq:Drhoataminimum}
\end{equation}
\end{linenomath}
for all $x_1<a^-$. Taking the limit $x_1 \nearrow a^{-}$ in~\eqref{eq:Drhoataminimum} and noting that $\dot{\varrho}$ is a continuous function we get that
\begin{linenomath}
\begin{equation}
    \dot{\rho}(a^{-}) \ge - \tan(\theta).
\label{eq:rho_am}
\end{equation}
\end{linenomath}
Since $\Dro(a^{-})=:\Drom$,  the left hand side in~\eqref{eq:rho_am} simplifies to $\Drom$. It follows from our hypothesis that the global compatibility condition is violated that the right hand side of~\eqref{eq:rho_am} is greater than $\Drom$. Thus, we get a contradiction.  Hence, our assumption that the $a^{-}$ configuration does not involve contact is false.

Since the $a^{-}$ configuration involves contact we need to move to line 7 of Algorithm~\ref{algo:TruePeelingAngle} for determining the configuration's true peeling angle. 
Noting that we have assumed $\DDro$ to be a continuous function and, from~\eqref{eq:DRhoPlusMinus}, that $\Drom$ is $\Dro$'s minimum value we get that $\DDro(a^-)=0$. 
This last result together with~\eqref{eq:Dl} and~\eqref{eq:curvature} implies that $k(a^{-})=0$.
The function $\varrho$'s property that it is a surjective function with the range $[-1,1]$ implies that $\Drom$ is negative. 
As a consequence of these last two implications we need to move to line 8 from line 7 in the algorithm. 
Since in forward peeling $\theta<\pi/2$ we get from line 8 that the true peeling angle in the $a^{-}$ configuration is naught, i.e., $\psi(a^-)=0$. 

In~\S\ref{sec:PeelingWithContact} we discussed that the equilibrium force corresponding to a configuration with de-adhered length $a$ is $\mathpzc{F}(\psi(a))$. 
Thus, we the $a^{-}$ configuration's equilibrium force is $\mathpzc{F}(0)$, which simplifies to $(2w)^{1/2}$. 

Recall that $(2w)^{1/2}$ is the function $\mathpzc{F}$'s maximum value. This fact in conjunction with $F^+$'s definition~\eqref{eq:Fmp} and the final result from the previous paragraph imply that $F^{+}=(2w)^{1/2}$.
\end{proof}

\begin{theorem}
  During backward peeling when the global compatibility condition~\eqref{eq:GlobalComp} is violated, specifically when $\theta > \pi - \tan^{-1}\pr{\Drop}$, the peel-initiation force achieves its lower bound, which is $(4+2w)^{1/2}-2$.
  \label{proof:backward_lb}
  \label{proof:upperbound}
\end{theorem}

We omit our proof for Theorem~\ref{proof:backward_lb}. Since it is quite similar to the one we provided for Theorem~\ref{proof:forward_ub}, except that in it we focus on the configuration with de-adhered length $a^{+}$ instead of the configuration with de-adhered length $a^{-}$.

  \begin{theorem}
  During backward peeling when the global compatibility condition is violated, specifically when
\begin{linenomath}
\begin{subequations}
\begin{align}
  \theta &> \pi - \tan^{-1}\pr{\Drop}
  \label{ineq:hypo}
  \\
  \intertext{the equilibrium force is always less than or equal to $\mathcal{F}\pr{\psi_{\rm lb}^{-}}$, where }
    \psi_{\rm lb}^{-}
    &:= \pi + \tan^{-1}\pr{\Drom} - \tan^{-1}\pr{\Drop}.
    \label{eq:psi_lb}
\end{align}
\end{subequations}
\end{linenomath}
  \end{theorem}

  \begin{proof}

  In \S\ref{sec:PeelingWithContact}  we discovered that during backward peeling when the global compatibility condition is violated  there  will exist some configurations that involve contact and others that do not (see Figures~\ref{fig:SinFContact}e and~\ref{fig:ZigFContact}e.) 

  When a configuration involves contact the true peeling angle is given by~\eqref{eq:psi_NC}.
  Since, by hypothesis, $\theta > \pi - \tan^{-1}\pr{\Drop}$ and, since  $\Dro(a)\ge \dot{\rho}^{-}$ by definition, $\tan^{-1}\pr{\Dro(a)} \ge  \tan^{-1}\pr{\Drom}$ it follows from~\eqref{eq:psi_NC} that during backward peeling when the configuration does not involve contact the true peeling angle is greater than   $\pi + \tan^{-1}\pr{\Drom} - \tan^{-1}\pr{\Drop}$, which is nothing but $\psi_{\rm lb}^{-}$.  

When a configuration involves contact then it follows from Algorithm~\ref{algo:TruePeelingAngle} that the true peeling angle is either equal to  $\pi$ (local contact) or to $\pi + \tan^{-1}\pr{\Dro(a)} - \tan^{-1}\pr{\Dro(c)}$ (non-local contact), where recall that  $a$ and $c$ are the abscissae of the peeling and contact fronts, respectively. Since we have assumed $\rho$ to be a non-constant function it follows that $\Drop>\Drom$ and hence that $\pi> \psi_{\rm lb}^{-}$.  It follows from the definitions of $\Drop$, $\Drom$, and the monotonicity of $\tan^{-1}$ that  $\pi+\tan^{-1}\pr{\Dro(a)} - \tan^{-1}\pr{\Dro(c)}\ge \psi_{\rm lb}^{-}$.
Thus, in the case of contact the true peeling angle is greater than or equal to $\psi_{\rm lb}^{-}$.

The deductions in the last two paragraphs can be summarized by saying that when~\eqref{ineq:hypo} holds the true peeling angle is always greater than or equal to $\psi_{\rm lb}^{-}$. In~\S\ref{sec:PeelingWithContact} we discussed that the equilibrium force corresponding to the true peeling angle $\psi(a)$ is always $\mathpzc{F}(\psi(a))$, irrespective of whether or not the configuration involves contact. Since $\mathpzc{F}$ is a monotonically decreasing function the last two statements imply that under~\eqref{ineq:hypo} the equilibrium peeling force is always less than  or equal to $\mathcal{F}\pr{\psi_{\rm lb}^{-}}=:F_{\rm ub}^{+}$.
\end{proof}

  For illustrating Theorem~\ref{proof:upperbound}, we mark  $\psi_{\rm lb}^-$ and ${F}_{\rm ub}^{+}$  for the case of backward peeling under~\eqref{ineq:hypo} on a simple wavy surface in Figures~\ref{fig:SinFContact}d and e, respectively, and on a complicated wavy surface in Figures~\ref{fig:ZigFContact}d and e, respectively.

  \section{Angle-independent optimal peel-off force}
  \label{sec:Optimal}

  In this section, we analyze the asymptotic value of the peel-off force $F^{+}$ when the substrate's aspect ratio $\alpha$, or its root-mean-square (RMS)\footnotemark~roughness,  becomes large. 
  This is equivalent to, e.g., the case where the substrate's surface's periodicity $\lambda$ becomes vanishingly small in comparison to its amplitude $A$. 

  \footnotetext{The RMS roughness of the substrate's surface is equal to $\alpha \pr{\int_0^1 \varrho(x_1)^2 \, dx_1}^{1/2}$.}

It follows from Algorithm~\ref{algo:TruePeelingAngle} that as $\alpha$ becomes large all configurations become contact configurations, irrespective of $a$ or $\varrho$, if $\theta$ is different from $\pi/2$; and if  $\theta=\pi/2$  none of the configurations involve contact. 

When $\theta = \pi/2 $, since none of the configurations involve contact, we can use the results given in \S\ref{sec:PeelingNoContact}. 
Specifically, using~\eqref{eq:psi_NC_pm} we get that  as $\alpha$ becomes large $\psi^{-}$ becomes vanishingly small, since in that limit $\Drom$ tends to negative infinity. 
Taking the limit $\psi^{-}\to 0$ in~\eqref{eq:Fmp_NC} we get the result that $F^{+}$ approaches its upper bound $(2 w)^{1/2}$ as $\alpha$ becomes large. 
This result is shown illustrated in Figure~\ref{fig:LargeA}. 

When $\theta<\pi/2$, the condition $\theta<-\tan^{-1}(\Drom)$ will inevitably get violated for a large enough $\alpha$. Thus, from Theorem~\ref{proof:forward_ub} we get that  $F^{+}$ will eventually equal its upper bound of $(2 w)^{1/2}$ as $\alpha$ becomes large. 

When $\theta > \pi/2 $, let us choose an $a$ for which $\dot{\varrho}(a)$ is negative. 
Now consider a sequence of configurations that all have that same $a$, $\varrho$, and $\theta>\pi/2$ but increasingly larger values of $\alpha$. 
It follows from Algorithm~\ref{algo:TruePeelingAngle} that when $\alpha$ becomes large enough all subsequent configurations will be of the non-local-contact type, and that the true peeling angle in all of them can be computed as $ \pi + \tan^{-1}\pr{\Dro(a)} - \tan^{-1}\pr{\Dro(c)}$. 
Recall that $c$ is the contact front's abscissa.
We would expect $c$ to vary between the configurations. 
However, it can be shown, again using Algorithm~\ref{algo:TruePeelingAngle}, that once the configurations become of the non-local-contact type the $c$ in them also remains fixed, and furthermore that the value of $\Dot{\varrho}$ at that $c$ is positive. 
Since $\dot{\rho}=\alpha \dot{\varrho}$ and $\dot{\varrho}$ is negative at $a$ and positive at $c$ as $\alpha$ becomes large $\tan^{-1}\pr{\Dro\pr{a}}$ and $\tan^{-1}\pr{\Dro\pr{c}}$ tend to $\mp \pi/2$, respectively, and, consequently, the true peeling angles tend to naught.
Recall from the discussion in \S\ref{sec:PeelingWithContact} that irrespective of whether or not a configuration involves contact, the equilibrium force in that configuration can always be computed as $\mathcal{F}(\psi(a))$, where $\psi(a)$ is the configuration's true peeling angle.
Thus, as $\alpha$ becomes large the equilibrium force in the sequence converges to $\mathcal{F}(0)=(2 w)^{1/2}$.
From~\eqref{eq:Fmp}  we know that the peel-off force for each geometry corresponding to a configuration in the sequence, namely that defined by  the profile $\alpha \varrho$ and the peeling angle $\theta$, is greater than or equal to that configuration's equilibrium force.
As we noted in the discussion immediately ensuing~\eqref{eq:Fmp}  the peel-off force, irrespective of  profile or peeling-angle, is bounded above by $(2 w)^{1/2}$.
It follows from the last three statements that for any fixed $\varrho$ and $\theta>\pi/2$ as  $\alpha$ becomes large the  peel-off force tends to its upper bound $(2w)^{1/2}$.

 %
\begin{figure}[t!]
   \centering
   \includegraphics[width=\textwidth]{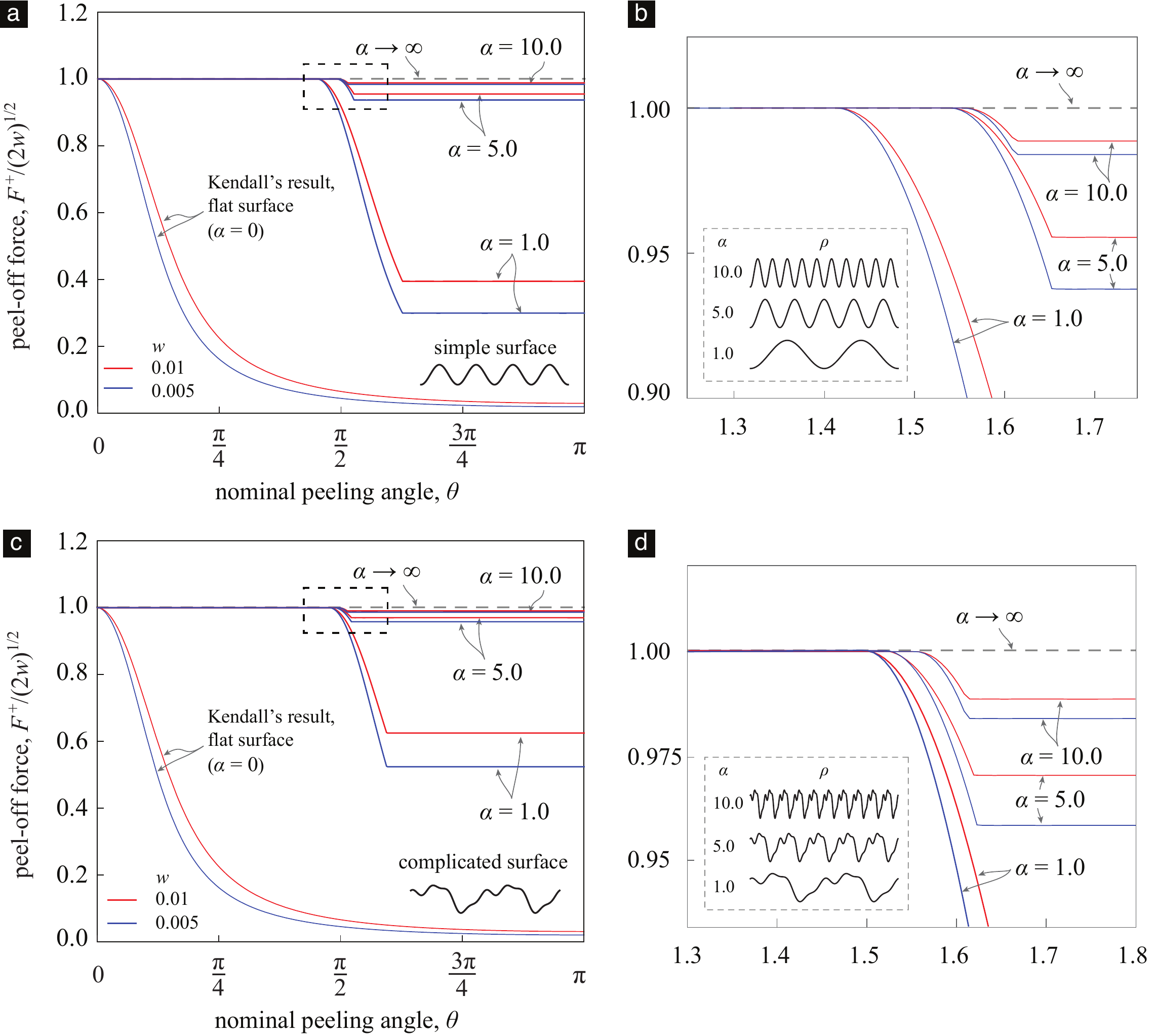}
   \caption{The plot of peel-off force $F^{+}$ as a function of nominal peeling angle $\theta$ for the peeling on (a) simple and (b) complicated  surfaces for a series of increasing surface's aspect ratio $\alpha$. The surface profiles, $\varrho$, are the same as those considered in Figures~\ref{fig:FNonContact}c and f, respectively.
   }
   \label{fig:LargeA}
 \end{figure}

  In summary, from the previous three paragraphs we have the important conclusion that as the surface roughness, $\alpha$, is increased the peel-off force $F^{+}$ tends to its upper bound $(2w)^{1/2}$.
  This happens independent of the surface's shape, $\varrho$, and  the peeling angle, $\theta$. We call this phenomenon angle-independent optimal adhesion.
  This phenomenon is quite interesting considering that adhesion on a flat surface is highly dependent on the peeling angle, and the optimal adhesion is only attained at a single peeling angle, namely for  $\theta=0$.

  We numerically computed the peel-off force $F^+$ for  the simple and complicated wavy surface shapes, which we first considered in~\S\ref{sec:PeelingNoContact}, for various peeling angles. For each surface shape and peeling angle we calculated $F^{+}$ for a sequence of geometries of increasing $\alpha$ values. The results of our calculations are shown in Figure~\ref{fig:LargeA}. As can be seen, at small $\alpha$ values, for instance $\alpha=1.0$ and $5.0$, the peel-off force ${F}^{+}$ depends strongly on $\theta$. However, as $\alpha$ increases the dependence of ${F}^{+}$ on $\theta$ becomes weak, such that at large $\alpha$ values, e.g.  $\alpha= 10.0$, the calculated peel-off forces appear to be essentially independent of the peeling angle. Finally, irrespective of the peeling angle or the surface shape, the calculated peel-off force values appears to approach $(2w)^{1/2}$ from below as $\alpha$ increases.

\section{Concluding remarks}
\label{sec:discussion}

We conclude this paper by briefly commenting on the effect  of bending strain energy on the peel-off force. 

Peng and Chen~\cite{peng2015peeling} investigated the peeling of an elastic film  on a sinusoidal surface. Their sinusoidal surface is the same our simple wavy surface that we first considered in \S\ref{sec:PeelingNoContact}, and is shown illustrated in Figure~\ref{fig:FNonContact}a. In their analysis they took into account the film's bending energy, where as in our model we only consider the  strain energy due to tension and ignore the strain energy due to bending.
They do not consider contact between the thin-film and the substrate in their analysis. Therefore, some insight into the effect of the bending strain energy can be garnered by comparing their analysis to the results that we present in  \S\ref{sec:PeelingNoContact} (peeling with no contact) when they are particularized to the case of simple wavy surface.

 In Peng and Chen's model the equilibrium peeling force, $F$, is related to the de-adhered length, $a$, as
 \begin{linenomath}
\begin{equation}
  F(a) = \cos(\psi({a})) - 1 + \left(\left( \cos(\psi({a})) - 1 \right)^2 + 2{w} - \frac{2\pi^4 \alpha^2\bar{h}^2 (1+\cos(4\pi{a})) }{3 (1+4\pi^2 \alpha^2 \sin^2(2\pi {a}))^{1/2}} \right)^{1/2},
  \label{eq:PengsEq}
\end{equation}
\end{linenomath}
where $\bar{h} = h/\lambda$. On particularizing the results in \S\ref{sec:PeelingNoContact} to the case of simple wavy surface our model predicts the same relation between the equilibrium force and de-adhered length as~\eqref{eq:PengsEq} except that in it there are no terms containing $\bar{h}$. That is, in our model's prediction the third terms from the left within the parenthesis of~\eqref{eq:PengsEq}, containing $\bar{h}$, is absent.
Since the third term scales with $\bar{h}$ as $\bar{h}^{2}$,  Peng and Chen's results converge to our results as $\bar{h} \to 0$. This fact can also be noted from Figure~\ref{fig:PengsEq}, in which we compare the  predictions from Peng and Chen's model with those from our model for the case of no-contact and peeling on a simple wavy surface for various $\bar{h}$  (in subfigure a) and $\theta$ (in subfigure (b)) values.
Even though the above comparison is only for a particular substrate profile, namely the simple wavy surface, we believe that it is reasonable to expect that the effect of the bending strain energy on the equilibrium forces will decrease with decreasing film thicknesses.

Another interesting feature of the above above comaparision that is revealed by the Figure~\ref{fig:PengsEq} and merits further investigation is that ignoring the bending strain energy seems to have no affect on the peel-off force. The peel-off force $F^{+}$ which we might recall is the supremum of all equilibrium forces on a given substrate profile and peeling angle.  Specifically, as per Figure~\ref{fig:PengsEq} the surpremum of the equilibrium forces predicted by our's as well as Peng and Chen's model appear to be the same. This numerical evidence prompts us conjecture that the bending strain energy, or the film thickness, will have no effect on a film's peel-off force irrespective of the substrate profile or peeling angle.

\begin{figure}[t!]
  \centering
  \includegraphics[width=\textwidth]{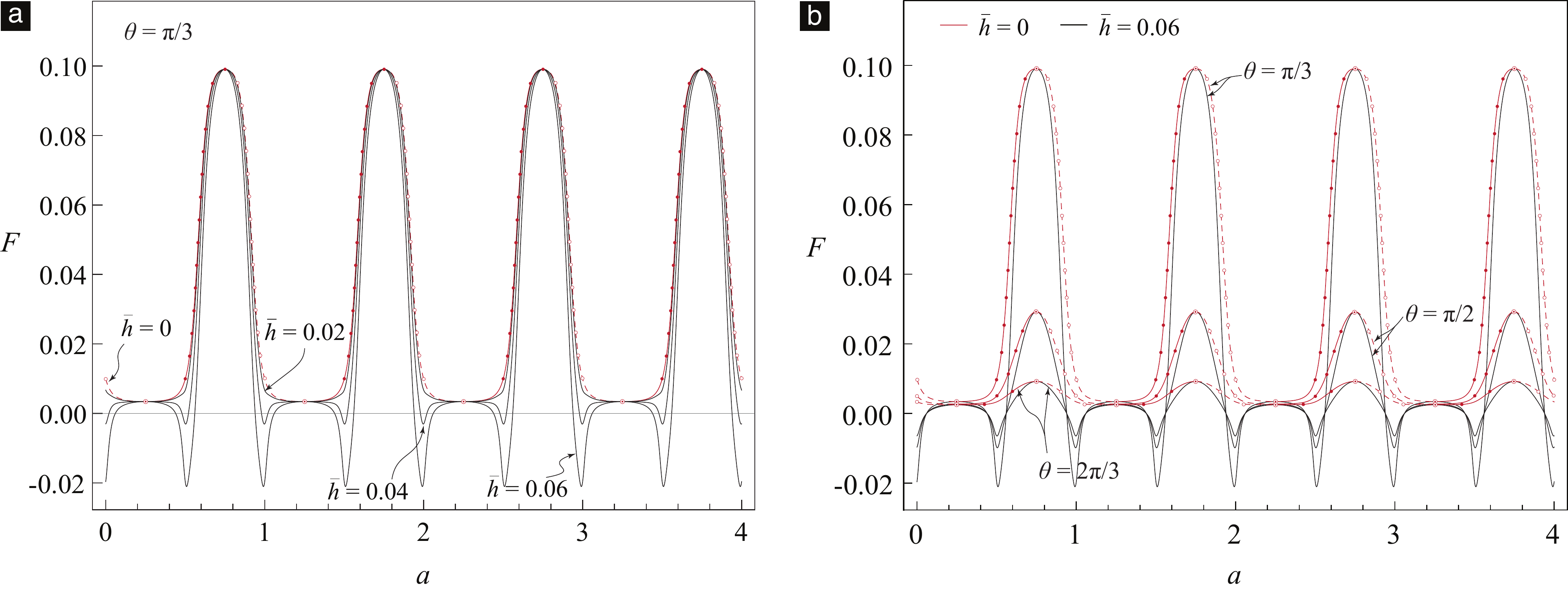}
  \caption{(a) The ${F}$--${a}$ plot for different values of thin film thicknesses $\bar{h}$ that considers the thin elastic film's bending energy for peeling on a sinusoidal surface at $\theta = \pi/3$. (b) The ${F}$--${a}$ plot for different values of nominal peeling angle $\theta$ for peeling a thin film with thickness $\bar{h} = 0$ (red curves) and $\bar{h} = 0.06$ (black curves) on a sinusoidal surface. For the sinusoidal surface in (a) and (b), we take $A = 0.25$, $\lambda = 1.0$, $\varrho(x_1) = -\cos(2\pi x_1)$, and ${w} = 0.005$. For $\bar{h}=0$, the stable, neutral, and unstable equilibrium state are marked with a solid dot, circled dot, and circle, respectively,  from the stability analysis of our model.}
  \label{fig:PengsEq}
\end{figure}

\section*{Acknowledgment}
The authors gratefully acknowledge support from the Office of Naval Research (Dr. Timothy Bentley) [Panther Program, grant number N000141812494] and the National Science Foundation [Mechanics of Materials and Structures Program, grant number 1562656]. Weilin Deng is partially supported by a graduate fellowship from the China Scholarship Council.

\appendix
\numberwithin{equation}{section}
\renewcommand{\theequation}{A.\arabic{equation}}
\renewcommand{\thefigure}{\arabic{figure}}
\section{Derivation of \texorpdfstring{$\delta u$}{} for peeling involving contact}
\label{sec:Appen_du}

\begin{figure}[t!]
  \centering
  \includegraphics[width=\textwidth]{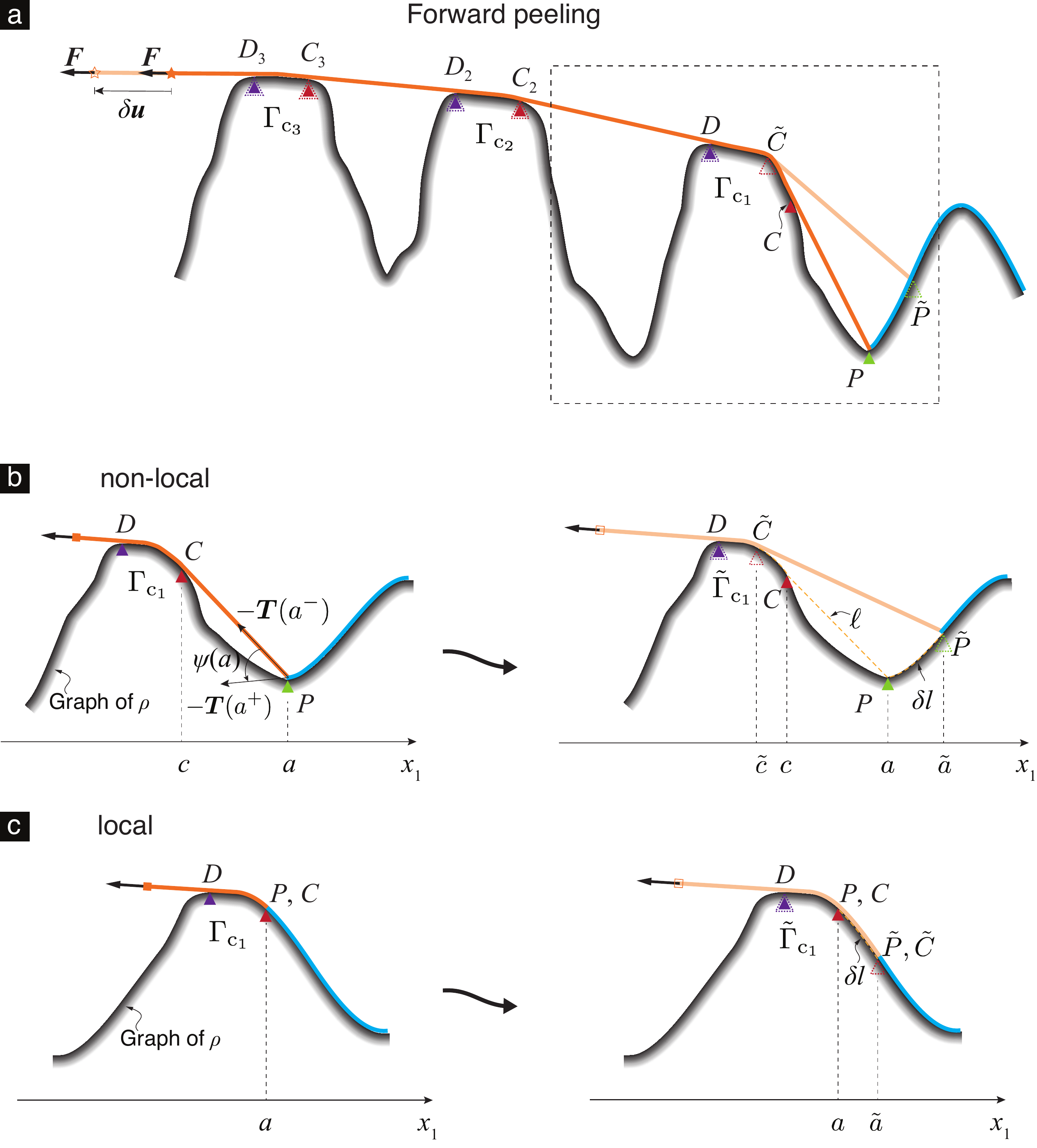}
  \caption{(a) The schematic of forward peeling involving contact that considers the peeling front $P$ moves to $\tilde{P}$ after an infinitesimal perturbation. The contact front accordingly changes from $C$ to $\tilde{C}$. (b) and (c) show non-local and local type contact, respectively.}
  \label{fig:ForwardPeeling}
\end{figure}
\begin{figure}[t!]
  \centering
  \includegraphics[width=\textwidth]{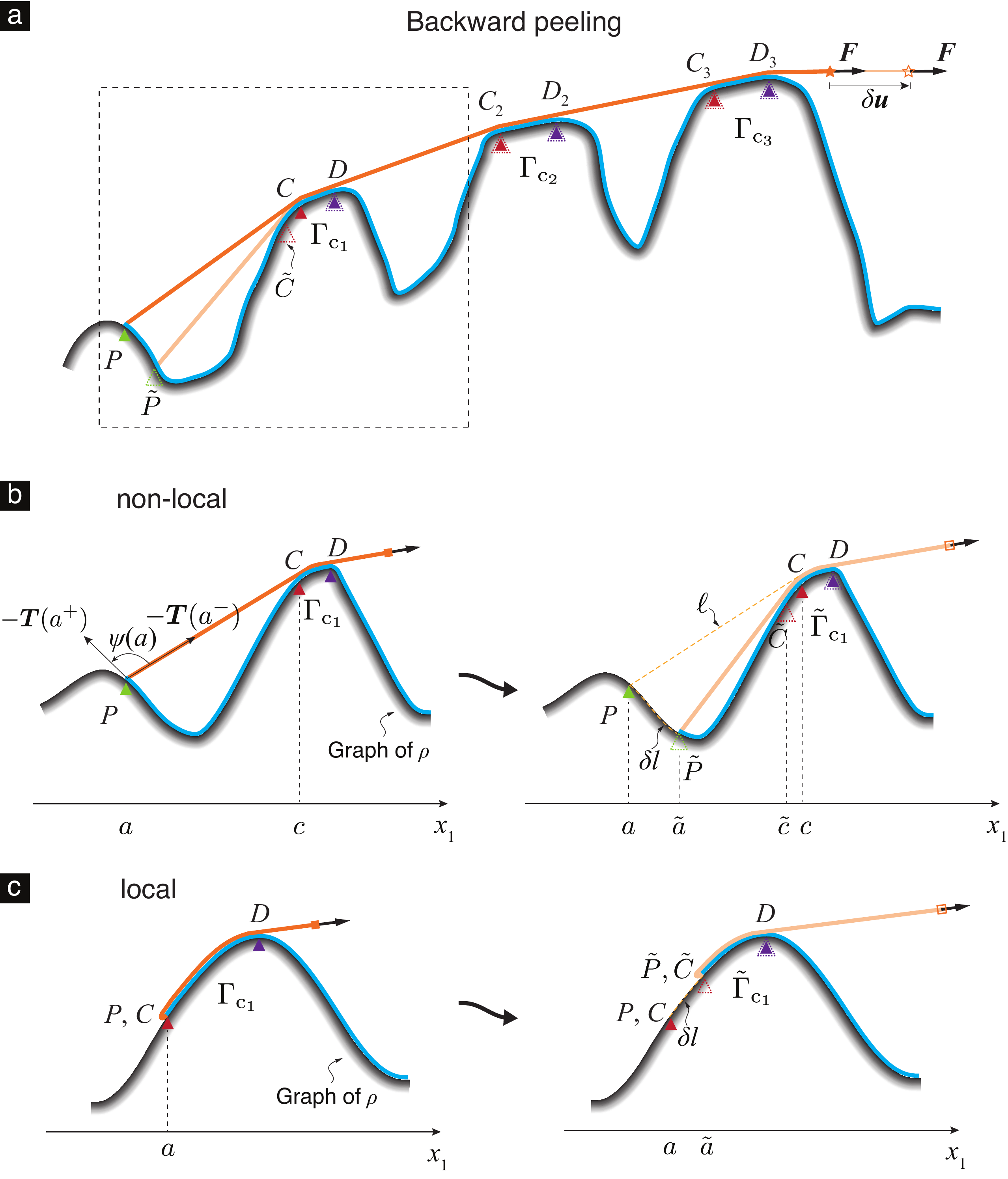}
  \caption{(a) The schematic of backward peeling involving contact that considers the peeling front $P$ moves to $\tilde{P}$ after an infinitesimal perturbation. The contact front accordingly changes from $C$ to $\tilde{C}$. (b) and (c) show non-local and local type contact, respectively.}
  \label{fig:BackwardPeeling}
\end{figure}

Configurations that involve contact during a generic forward peeling experiment are shown in Figure~\ref{fig:ForwardPeeling}, and similar configurations from a backward peeling experiment are shown in Figure~\ref{fig:BackwardPeeling}. For concreteness we focus on the forward peeling experiment, however, our remarks apply to the backward peeling experiment as well. The  peeled part and the remainder of the adhered part of the thin film in the configuration $\conf$ are shown, respectively, in dark orange and blue in Figure~\ref{fig:ForwardPeeling}a. We perturb $\conf$ slightly to obtain the nearby configuration $\confp$. The peeled part of the thin film in  $\confp$ is shown in light orange in  Figure~\ref{fig:ForwardPeeling}a. Recall that $P$ and $C$ denote the peeling and contact fronts  in $\conf$. We denote the peeling and contact fronts in $\confp$ as $\tilde{P}$ and $\tilde{C}$.

The quantity that we aim to compute, namely, $\delta u$, is related to the thin film's kinematics. The thin film's kinematics take place in $\mathcal{E}$. However, since the thin film's kinematics do not vary in  the $\physe_3$ direction, $\delta u$ can be computed  by only analyzing the kinematics that take place in, say, $\mathcal{E}$'s $x_1$-$x_2$ plane. Therefore, we will be focusing all our analysis related to computing $\delta u$  only to $\mathcal{E}$'s  $x_1$-$x_2$ plane. In order to avoid the introduction of more symbols,  we will use the same symbols that we introduced to refer to subsets of  $\mathcal{E}$ to refer to the quantities that result from the projection of those subsets into the  $x_1$-$x_2$ plane. For instance. The peeling and contact fronts, $P$ and $C$, are  line segments in $\mathcal{E}$. However, we will be denoting their projections  into the $x_1$--$x_2$ plane, which are points in  $\mathcal{E}$, also as $P$ and $C$, and refer to them as the peeling and contact point, respectively. Similarly, we denote the   peeling and contact points in $\confp$ as $\tilde{P}$ and $\tilde{C}$.

 We consider a general contact scenario in which the contact region $\Gamma_c\in \mathcal{D}$ consists of several disjoint contact patches, $\Gamma_{c_1}$,$\Gamma_{c_2}$, etc. In Figure~\ref{fig:ForwardPeeling}a we mark the  regions on the substrate which mate with  these contact patches in $\conf$ as $\u{x}\pr{\Gamma_{\rm c_1}}$, $\u{x}\pr{\Gamma_{\rm c_2}}$,\ldots, $\u{x}\pr{\Gamma_{\rm c_n}}$. We call $\u{x}\pr{\Gamma_{\rm c_1}}$ the first contact region, and $\u{x}\pr{\Gamma_{\rm c_n}}$  the last contact region. By definition, the point on the substrate where the first contact region begins is the contact point $C$. We mark the  point where the first contact region ends as $D$. Similarly, we mark the points where the $i^{\rm th}$, $i=2,\ldots, n $, contact region begins and ends as $C_i$ and $D_i$, respectively. We denote the location of the thin film's terminal end in $\conf$ as $X$.

 As we perturb the thin film's configuration from $\conf$ to  $\confp$, the peeling front moves from $P$ to $\tilde{P}$, the contact front  moves from $C$ to $\tilde{C}$, and the terminal end from $X$ to $\tilde{X}$. Interestingly, however, none of the other features that define the thin film's geometry move  during this perturbation. Specifically, in $\confp$ the first contact region still ends at  $D$, and all the other contact regions  still begin and end at the same points that they did in $\conf$.

The length of the peeled part of the  film in $\conf$ can be computed as the sum of the lengths of the line segment $\widebar{CP}$, the arc $\arc{CD}$, the line segment $\widebar{DC_{2}}$, the arcs $\arc{C_iD_i}$, $i=2,\ldots, n$, and the line segments $\widebar{D_iC_{i+1}}$, $i=2,\ldots, n-1$, and, finally, the line segment $\widebar{D_nX}$. Based on the discussion in the previous paragraph, the length of the peeled part  in $\confp$ is equal to the sum of the lengths of the line segment $\widebar{\t{P}\t{C}}$, the arc $\arc{\t{C} D}$, and, as before, the line segment $\widebar{D C_{2}}$, the arcs $\arc{C_i D_i}$, and the line segments $\widebar{D_iC_{i+1}}$, and, finally, the line segment $\overline{D_n \tilde{X}}$.  The difference in the length of the peeled part between  $\conf$ and $\confp$, therefore, is
\begin{linenomath}
\begin{equation}
\overline{D_n \tilde{X}}-\overline{D_n X}+
\arc{\t{C}D}-\arc{CD}+
\overline{\t{P}\t{C}}-\overline{PC}.
\notag
\end{equation}
\end{linenomath}
The length difference can, alternately, also be  computed as $\delta l(1+\varepsilon)$. Equating these two expressions for the length difference and noting that  $\overline{D_n \tilde{X}}-
\overline{D_n X}$ is in fact  $\delta u$ we get that
\begin{linenomath}
\begin{equation}
    \delta u
    =\delta l(1+\varepsilon)-\pr{\arc{\t{C}D}-\arc{CD}}
    -\overline{\t{P}\t{C}}
    +\overline{PC}.
\label{eq:LengthDifference}
\end{equation}
\end{linenomath}

The term $\arc{\t{C}D}-\arc{CD}$ in~\eqref{eq:LengthDifference} can be computed as
\begin{linenomath}
\begin{equation}
 \arc{\t{C}D}-\arc{CD}
  = \pm \int^{\tilde{c}}_c\pr{1+\Dro(x_1)^2}^{1/2} \, dx_1 , \label{eq:arcCprime-C}
\end{equation}
\end{linenomath}
where the plus sign is for the case of forward peeling, while the minus sign is for the case of backward peeling.

The points $P$, $C$, $\tilde{P}$, and $\tilde{C}$ are shown marked in, e.g.,  Figure~\ref{fig:ForwardPeeling}a. From the definitions of peeling and contact fronts it follows that the coordinates of the points $P$ and $C$ are, respectively,  $(a,\rho(a))$ and $(c,\rho(c))$. By denoting the  abscissa of  the points $\tilde{P}$ and $\tilde{C}$ as $\tilde{a}$ and $\tilde{c}$ it follows for similar reasons that these points' coordinates are $(\tilde{a},\rho(\tilde{a}))$ and $(\tilde{c},\rho(\tilde{c})$, respectively. The line segments $\widebar{PC}$ and $\widebar{\t{P}\t{C}}$ are tangent to the graph of $\rho$ at $C$ and $\t{C}$, respectively. (These aspects of the film's geometry are especially clear in Figure~\ref{fig:ForwardPeeling}b.) Using this information it can be shown that
  \begin{linenomath}
  \begin{subequations}
  \begin{align}
    \rho\pr{a}-\rho\pr{c}
    &=   \Dro\pr{c}(a-c), \label{eq:TangentAtC1} \\
    \rho\pr{\tilde{a}}-\rho\pr{\tilde{c}}
    &=  \Dro\pr{\tilde{c}}(\tilde{a}-\tilde{c}). \label{eq:TangentAtCprime1}
  \end{align}
  \label{eq:TangentAtCCprime1}
  \end{subequations}
  \end{linenomath}
Using the  knowledge of $P$, $C$, $\t{P}$, and $\t{C}$'s coordinates, and~\eqref{eq:TangentAtCCprime1} it can be shown that
\begin{linenomath}
\begin{align}
    \overline{PC} &= \pr{1+\Dro(c)^2}^{1/2}\lvert c-a\rvert,
    \label{eq:lineC-P}\\
  \overline{\t{P}\t{C}} &= \pr{1+\Dro(\t{c})^2}^{1/2}\lvert \t{c}-\t{a}\rvert.
  \label{eq:linetC-tP}
\end{align}
\end{linenomath}

Substituting $\arc{\t{C}D}-\arc{CD}$, $\overline{PC}$, and $\overline{\t{P}\t{C}}$ in~\eqref{eq:LengthDifference} with the right hand sides of~\eqref{eq:arcCprime-C},~\eqref{eq:lineC-P}, and~\eqref{eq:linetC-tP},  respectively; then writing $\t{a}$ as $a+\delta a$, and $c$ as the right hand side of~\eqref{eq:DeltacOnDeltaa} in the resulting equation;  and then, finally, expanding the resulting equation in series of $\delta a$, we get that as $\delta a \to 0$
\begin{linenomath}
\begin{equation}
  \begin{aligned}
    \delta {u} =
    & \left(1+\varepsilon - \cos(\psi(a)) \right) \dot{l}({a};\rho) \delta{a} \\
    & + \frac{1}{2}\left( (1+\varepsilon)  \ddot{l}({a};\rho) - \frac{\dot{l}({a};\rho)^2 }{\ell} \sin^2(\psi(a)) + \frac{\pr{\sin(\psi(a))-\cos(\psi(a))\Dro(a)}\DDro(a)}{\Dl(a;\rho)} \right) (\delta {a})^2 + o\pr{(\delta a)^2},
  \end{aligned}
  \label{eq:du_C1}
\end{equation}
\end{linenomath}
where  $\ell$ is an alias for $\overline{PC}$. We introduce this new symbol so as to make some of the results that derive from~\eqref{eq:du_C1} appear more compact. The result~\eqref{eq:du_C1} applies to both forward as well as backward peeling. In arriving at~\eqref{eq:du_C1} we used~\eqref{eq:PLDerivatives}, and the result that
\begin{linenomath}
\begin{equation}
\cos\pr{\psi(a)}=\pm \frac{1+\Dro(a)\Dro(c)}{\pr{\pr{1+\Dro(a)^2}\pr{1+\Dro(c)^2}}^\frac{1}{2}},
\label{eq:CosPsia}
\end{equation}
\end{linenomath}
where the plus sign is for the case of forward peeling, while the minus sign is for the case of backward peeling. The result~\eqref{eq:CosPsia} follows from Algorithm~\ref{algo:TruePeelingAngle}.


\subsection{The asymptotic behavior of \texorpdfstring{$\delta c$}{} as \texorpdfstring{$\delta a\to 0$}{}}


Expressing $\t{ a}$ as $a+\delta a$ and $\t{c}$ as $c+f\pr{\delta a}$, where $f$ is some real valued analytic function over $\mathbb{R}$,  in~\eqref{eq:TangentAtCprime1}, and then expanding both sides of the resulting equation in series of $\delta a$, we get that as $\delta a\to 0$
\begin{linenomath}
\begin{equation}
\left((c-a) \Dro(c)+\rho (a)-\rho
   (c)\right)+
   \left((c-a) \dot{f}(0) \DDro(c)+\Dro(a)-\Dro(c)\right)\delta a+o\pr{\delta a}=0.
\label{eq:AsymExpansion}
\end{equation}
\end{linenomath}
In arriving at~\eqref{eq:AsymExpansion} we made use of the identity that $f(0)=0$, which is a consequence of the requirement that  $\t{c}\to c$ as  $\delta a\to 0$. It follows from~\eqref{eq:TangentAtC1} and~\eqref{eq:AsymExpansion} that
\begin{linenomath}
\begin{equation}
\delta c=\frac{\Dro(a)-\Dro(c)}{\DDro(c)(a-c)} \delta a+o\pr{\delta a}.
\label{eq:DeltacOnDeltaa}
\end{equation}
\end{linenomath}

\bibliography{WavyPeelingRefs}

\begin{thebibliography}{10}

\bibitem{barthlott1997purity}
Wilhelm Barthlott and Christoph Neinhuis.
\newblock Purity of the sacred lotus, or escape from contamination in
  biological surfaces.
\newblock {\em Planta}, 202(1):1--8, 1997.

\bibitem{wen2014biomimetic}
Li~Wen, James~C Weaver, and George~V Lauder.
\newblock Biomimetic shark skin: design, fabrication and hydrodynamic function.
\newblock {\em Journal of Experimental Biology}, 217(10):1656--1666, 2014.

\bibitem{fuller1975effect}
KNG Fuller and David Tabor.
\newblock The effect of surface roughness on the adhesion of elastic solids.
\newblock {\em Proceedings of the Royal Society of London. A. Mathematical and
  Physical Sciences}, 345(1642):327--342, 1975.

\bibitem{levins1995impact}
John~M Levins and T~Kyle Vanderlick.
\newblock Impact of roughness on the deformation and adhesion of a rough metal
  and smooth mica in contact.
\newblock {\em The Journal of Physical Chemistry}, 99(14):5067--5076, 1995.

\bibitem{quon1999measurement}
RA~Quon, RF~Knarr, and TK~Vanderlick.
\newblock Measurement of the deformation and adhesion of rough solids in
  contact.
\newblock {\em The Journal of Physical Chemistry B}, 103(25):5320--5327, 1999.

\bibitem{rabinovich2000adhesion}
Yakov~I Rabinovich, Joshua~J Adler, Ali Ata, Rajiv~K Singh, and Brij~M Moudgil.
\newblock Adhesion between nanoscale rough surfaces: Ii. measurement and
  comparison with theory.
\newblock {\em Journal of Colloid and Interface Science}, 232(1):17--24, 2000.

\bibitem{briggs1977effect}
GAD Briggs and BJ~Briscoe.
\newblock The effect of surface topography on the adhesion of elastic solids.
\newblock {\em Journal of Physics D: Applied Physics}, 10(18):2453, 1977.

\bibitem{kesari2010role}
Haneesh Kesari, Joseph~C Doll, Beth~L Pruitt, Wei Cai, and Adrian~J Lew.
\newblock Role of surface roughness in hysteresis during adhesive elastic
  contact.
\newblock {\em Philosophical Magazine \& Philosophical Magazine Letters},
  90(12):891--902, 2010.

\bibitem{kesari2011effective}
Haneesh Kesari and Adrian~J Lew.
\newblock Effective macroscopic adhesive contact behavior induced by small
  surface roughness.
\newblock {\em Journal of the Mechanics and Physics of Solids},
  59(12):2488--2510, 2011.

\bibitem{deng2019depth}
Weilin Deng and Haneesh Kesari.
\newblock Depth-dependent hysteresis in adhesive elastic contacts at large
  surface roughness.
\newblock {\em Scientific Reports}, 9(1):1639, 2019.

\bibitem{deng2019effect}
Weilin Deng and Haneesh Kesari.
\newblock Effect of machine stiffness on interpreting contact
  force--indentation depth curves in adhesive elastic contact experiments.
\newblock {\em Journal of the Mechanics and Physics of Solids}, 131:404--423,
  2019.

\bibitem{deng2017molecular}
Weilin Deng and Haneesh Kesari.
\newblock Molecular statics study of depth-dependent hysteresis in nano-scale
  adhesive elastic contacts.
\newblock {\em Modelling and Simulation in Materials Science and Engineering},
  25(5):055002, 2017.

\bibitem{michel1999effective}
Jean-Claude Michel, Herv{\'e} Moulinec, and P~Suquet.
\newblock Effective properties of composite materials with periodic
  microstructure: a computational approach.
\newblock {\em Computer methods in applied mechanics and engineering},
  172(1-4):109--143, 1999.

\bibitem{ericksen2012homogenization}
Jerry~L Ericksen, David Kinderlehrer, Robert Kohn, and J-L Lions.
\newblock {\em Homogenization and effective moduli of materials and media},
  volume~1.
\newblock Springer Science \& Business Media, 2012.

\bibitem{castaneda2002second}
Pedro~Ponte Castaneda.
\newblock Second-order homogenization estimates for nonlinear composites
  incorporating field fluctuations: I--theory.
\newblock {\em Journal of the Mechanics and Physics of Solids}, 50(4):737--757,
  2002.

\bibitem{popov2017strength}
Valentin~L Popov, Roman Pohrt, and Qiang Li.
\newblock Strength of adhesive contacts: Influence of contact geometry and
  material gradients.
\newblock {\em Friction}, 5(3):308--325, 2017.

\bibitem{ciavarella2019role}
Michele Ciavarella, J~Joe, Antonio Papangelo, and JR~Barber.
\newblock The role of adhesion in contact mechanics.
\newblock {\em Journal of the Royal Society Interface}, 16(151):20180738, 2019.

\bibitem{rivlin1944the}
R.S. Rivlin.
\newblock The effective work of adhesion.
\newblock {\em Paint Technol.}, 9:215--218, 1944.

\bibitem{kendall1975thin}
K.~Kendall.
\newblock Thin-film peeling-the elastic term.
\newblock {\em Journal of Physics D: Applied Physics}, 8(13):1449, 1975.

\bibitem{johnson1971surface}
KL~Johnson, K~Kendall, and AD~Roberts.
\newblock Surface energy and the contact of elastic solids.
\newblock {\em Proceedings of the Royal Society of London A: Mathematical,
  Physical and Engineering Sciences}, 324(1558):301--313, 1971.

\bibitem{helmet}
3d printed helmet.

\bibitem{feng2007competing}
Xue Feng, Matthew~A Meitl, Audrey~M Bowen, Yonggang Huang, Ralph~G Nuzzo, and
  John~A Rogers.
\newblock Competing fracture in kinetically controlled transfer printing.
\newblock {\em Langmuir}, 23(25):12555--12560, 2007.

\bibitem{pesika2007peel}
Noshir~S Pesika, Yu~Tian, Boxin Zhao, Kenny Rosenberg, Hongbo Zeng, Patricia
  McGuiggan, Kellar Autumn, and Jacob~N Israelachvili.
\newblock Peel-zone model of tape peeling based on the gecko adhesive system.
\newblock {\em The Journal of Adhesion}, 83(4):383--401, 2007.

\bibitem{chen2008pre}
Bin Chen, Peidong Wu, and Huajian Gao.
\newblock Pre-tension generates strongly reversible adhesion of a spatula pad
  on substrate.
\newblock {\em Journal of The Royal Society Interface}, 6(35):529--537, 2008.

\bibitem{molinari2008peeling}
Alain Molinari and Guruswami Ravichandran.
\newblock Peeling of elastic tapes: effects of large deformations,
  pre-straining, and of a peel-zone model.
\newblock {\em The Journal of Adhesion}, 84(12):961--995, 2008.

\bibitem{begley2013peeling}
Matthew~R Begley, Rachel~R Collino, Jacob~N Israelachvili, and Robert~M
  McMeeking.
\newblock Peeling of a tape with large deformations and frictional sliding.
\newblock {\em Journal of the Mechanics and Physics of Solids},
  61(5):1265--1279, 2013.

\bibitem{gialamas2014peeling}
Panayiotis Gialamas, Benjamin V{\"o}lker, Rachel~R Collino, Matthew~R Begley,
  and Robert~M McMeeking.
\newblock Peeling of an elastic membrane tape adhered to a substrate by a
  uniform cohesive traction.
\newblock {\em International Journal of Solids and Structures},
  51(18):3003--3011, 2014.

\bibitem{menga2018multiple}
N.~Menga, L.~Afferrante, NM~Pugno, and G.~Carbone.
\newblock The multiple v-shaped double peeling of elastic thin films from
  elastic soft substrates.
\newblock {\em Journal of the Mechanics and Physics of Solids}, 113:56--64,
  2018.

\bibitem{kim1988elasto}
Kyung-Suk Kim and Junglhl Kim.
\newblock Elasto-plastic analysis of the peel test for thin film adhesion.
\newblock {\em Journal of Engineering Materials and Technology},
  110(3):266--273, 1988.

\bibitem{kinloch1994peeling}
A.J. Kinloch, C.C. Lau, and J.G. Williams.
\newblock The peeling of flexible laminates.
\newblock {\em International Journal of Fracture}, 66(1):45--70, 1994.

\bibitem{wei1998interface}
Yueguang Wei and John~W. Hutchinson.
\newblock Interface strength, work of adhesion and plasticity in the peel test.
\newblock {\em International Journal of Fracture}, 93(1):315--333, 1998.

\bibitem{loukis1991effects}
M.J. Loukis and N.~Aravas.
\newblock The effects of viscoelasticity in the peeling of polymeric films.
\newblock {\em The Journal of Adhesion}, 35(1):7--22, 1991.

\bibitem{afferrante2016ultratough}
L.~Afferrante and G.~Carbone.
\newblock The ultratough peeling of elastic tapes from viscoelastic substrates.
\newblock {\em Journal of the Mechanics and Physics of Solids}, 96:223--234,
  2016.

\bibitem{kendall1975control}
K.~Kendall.
\newblock Control of cracks by interfaces in composites.
\newblock {\em Proceedings of the Royal Society of London A: Mathematical,
  Physical and Engineering Sciences}, 341(1627):409--428, 1975.

\bibitem{ghatak2004peeling}
Animangsu Ghatak, L~Mahadevan, Jun~Young Chung, Manoj~K Chaudhury, and Vijay
  Shenoy.
\newblock Peeling from a biomimetically patterned thin elastic film.
\newblock {\em Proceedings of the Royal Society of London A: Mathematical,
  Physical and Engineering Sciences}, 460(2049):2725--2735, 2004.

\bibitem{chung2005roles}
Jun~Young Chung and Manoj~K Chaudhury.
\newblock Roles of discontinuities in bio-inspired adhesive pads.
\newblock {\em Journal of the Royal Society Interface}, 2(2):55--61, 2005.

\bibitem{xia2012toughening}
Shuman Xia, Laurent Ponson, Guruswami Ravichandran, and Kaushik Bhattacharya.
\newblock Toughening and asymmetry in peeling of heterogeneous adhesives.
\newblock {\em Physical Review Letters}, 108(19):196101, 2012.

\bibitem{xia2013adhesion}
S.M. Xia, L.~Ponson, G.~Ravichandran, and K.~Bhattacharya.
\newblock Adhesion of heterogeneous thin films-{I}: Elastic heterogeneity.
\newblock {\em Journal of the Mechanics and Physics of Solids}, 61(3):838--851,
  2013.

\bibitem{xia2015adhesion}
SM~Xia, Laurent Ponson, Guruswami Ravichandran, and Kaushik Bhattacharya.
\newblock Adhesion of heterogeneous thin films ii: Adhesive heterogeneity.
\newblock {\em Journal of the Mechanics and Physics of Solids}, 83:88--103,
  2015.

\bibitem{zhao2013improvement}
Hong-Ping Zhao, Yecheng Wang, Bing-Wei Li, and Xi-Qiao Feng.
\newblock Improvement of the peeling strength of thin films by a bioinspired
  hierarchical interface.
\newblock {\em International Journal of Applied Mechanics}, 5(02):1350012,
  2013.

\bibitem{ghatak2014peeling}
Animangsu Ghatak.
\newblock Peeling off an adhesive layer with spatially varying topography and
  shear modulus.
\newblock {\em Physical Review E}, 89(3):032407, 2014.

\bibitem{peng2015peeling}
Zhilong Peng and Shaohua Chen.
\newblock Peeling behavior of a thin-film on a corrugated surface.
\newblock {\em International Journal of Solids and Structures}, 60:60--65,
  2015.

\bibitem{maugis:book}
D.~Maugis.
\newblock {\em Contact adhesion and rupture of elastic solids}.
\newblock Solid State Sciences. Springer, 2000.

\end{thebibliography}

\end{document}